\documentclass[USenglish]{lipics-v2016}

\usepackage[utf8]{inputenc}

\usepackage{microtype}

\let\oldenumerate\enumerate\let\oldendenumerate\endenumerate
\usepackage{enumerate}

\let\enumerate\oldenumerate\let\endenumerate\oldendenumerate

\usepackage{amssymb}
\usepackage{stmaryrd}
\usepackage{amsmath}
\usepackage{mathrsfs}
\usepackage{colonequals}
\usepackage{url}
\usepackage{bm}
\usepackage[noadjust,compress,sort]{cite}

\makeatletter
\let\oldbegintheorem\@begintheorem                                      
\let\oldendtheorem\@endtheorem

\makeatother

\usepackage{apxproof}
\usepackage{xcolor}

\makeatother

\newcommand{\dom}{\lambda}
\newcommand{\defeq}{\colonequals}

\newcommand{\card}[1]{\left|#1\right|}

\newcommand{\calI}{\mathcal{I}}

\newcommand{\calS}{\mathcal{S}}

\newcommand\AND{\mathsf{and}}
\newcommand\OR{\mathsf{or}}
\newcommand\NOT{\mathsf{not}}

\newcommand\VARS{\mathsf{Vars}}
\newcommand\LEAVES{\mathsf{Leaves}}

\newcommand\IT{\mathrm{VarT}}

\newcommand\SINT{\mathrm{SINT}}
\newcommand\SDISJ{\mathrm{SDISJ}}

\renewcommand\root{\mathrm{root}}
\newcommand\spl{\mathrm{Split}}

\newcommand\psw{\mathrm{psw}}
\newcommand\tsw{\mathrm{tsw}}

\newcommand\first{\mathit{first}}

\newcommand\pw{\mathrm{pw}}
\newcommand\tw{\mathrm{tw}}

\newcommand\degree{\mathrm{degree}}

\newcommand\UNF{\mathrm{Unj}}

\newcommand{\NN}{\mathbb{N}}

\newcommand{\arity}{\mathrm{arity}}

\newcommand{\children}{\mathsf{children}}

\newcommand{\var}{\mathsf{var}}

\newcommand{\out}{\mathsf{output}}

\newcommand{\myparagraph}[1]{\subparagraph*{#1.}}

\newcommand{\ucqneq}{\mathrm{UCQ}^{\neq}}

\newcommand{\p}{\mathrm{p}}

\newcommand\restr[2]{{%
  \kern-\nulldelimiterspace %
  #1 %
  _{|#2} %
  }}

\theoremstyle{plain}
\newtheoremrep{theorem}{Theorem}
\newtheoremrep{corollary}[theorem]{Corollary}
\newtheoremrep{lemma}[theorem]{Lemma}
\newtheoremrep{proposition}[theorem]{Proposition}
\newtheoremrep{claim}[theorem]{Claim}
\newtheoremrep{remark}[theorem]{Remark}
\newtheoremrep{observation}[theorem]{Observation}
\newtheoremrep{conjecture}[theorem]{Conjecture}
\newtheoremrep{question}[theorem]{Question}
\newtheoremrep{note}[theorem]{Note}

\newtheorem{result}{Result} %

\title{Connecting Width and Structure\protect\\in Knowledge Compilation
(Extended Version)}
\titlerunning{Connecting Width and Structure in Knowledge Compilation
(Extended Version)}
\DOIPrefix{}

\author[1]{Antoine Amarilli}
\author[1,3]{Mikaël Monet}
\author[1,2,3]{Pierre Senellart}

\affil[1]{LTCI, Télécom ParisTech, Université Paris-Saclay}
\affil[2]{DI ENS, ENS, CNRS, PSL Research University; Paris, France}
\affil[3]{Inria Paris; Paris, France}

\hypersetup{
    colorlinks,
    linkcolor={red!50!black},
    citecolor={blue!50!black},
    urlcolor={blue!30!black}
}

\renewcommand{\phi}{\varphi}
\renewcommand{\epsilon}{\varepsilon}
\renewcommand{\leq}{\leqslant}
\renewcommand{\geq}{\geqslant}

\begin{document}

\maketitle

\begin{abstract}
  Several query evaluation tasks can be done via \emph{knowledge
compilation}: the query result is compiled as a \emph{lineage circuit}
from which the answer can be determined. For such tasks, it is important
to leverage some width parameters of the circuit, such as bounded
treewidth or pathwidth, to convert the circuit to structured classes,
e.g., deterministic structured NNFs (d-SDNNFs) or OBDDs. In this work, we
show how to connect the width of circuits to the size of their structured
representation, through upper and lower bounds. For the upper bound, we
show how bounded-treewidth circuits can be converted to a d-SDNNF, in
time linear in the circuit size. Our bound, unlike existing
results, is constructive and only singly exponential in the treewidth. We show a related lower
bound on monotone DNF or CNF formulas, assuming a constant bound on the
arity (size of clauses) and degree (number of occurrences of each
variable). Specifically, any d-SDNNF (resp., SDNNF) for such a DNF (resp., CNF) must be of
exponential size in its treewidth; and the same holds for pathwidth when
compiling to OBDDs. Our lower bounds, in contrast with most previous work, apply to
\emph{any} formula of this class, not just a well-chosen family. Hence,
for our language of DNF and CNF, pathwidth and treewidth respectively
characterize the efficiency of compiling to OBDDs and \mbox{(d-)SDNNFs}, that is,
compilation is singly exponential in the width parameter. We
conclude by applying our lower bound results to the task of query evaluation.

\end{abstract}

\section{Introduction}
\label{sec:introduction}
Uncertainty and errors in data can be modeled using
\emph{probabilistic
databases}~\cite{suciu2011probabilistic}, annotating every tuple with a
probability of existence. 
Query evaluation on probabilistic databases must then handle the uncertainty by
computing the probability that each query result holds.
A common technique to evaluate queries on probabilistic databases is
the \emph{intensional approach}:
first compute a representation of the \emph{lineage} 
of the
query
on the database, which intuitively describes how the query depends on the
possible database facts;
then use this lineage to compute probabilities efficiently.
Specifically, the lineage can be computed as a
\emph{circuit}~\cite{jha2012tractability}, and efficient probability computation
can be achieved by restricting to tractable circuit classes via
\emph{knowledge compilation}. Thus, 
to evaluate queries on probabilistic databases, we can use 
knowledge compilation algorithms to
translate circuits to tractable classes; conversely, lower bounds in knowledge compilation can identify the
limits of the intensional approach.

In this paper, we study the relationship between two kinds of tractable circuit
classes in knowledge compilation: \emph{width-based} classes, specifically, 
bounded-treewidth and bounded-pathwidth circuits; and \emph{structure-based}
classes, specifically, OBDDs (ordered binary decision diagrams
\cite{bryant1992symbolic}, following a \emph{variable order}) and d-SDNNFs
(structured deterministic decomposable negation normal
forms~\cite{pipatsrisawat2008new}, following a
\emph{v-tree}).
Circuits of bounded treewidth can be
obtained as a result of practical query
evaluation~\cite{jha2010bridging,amarilli2015provenance, amarilli2017combined}, whereas OBDDs and
d-DNNFs have been studied to show theoretical characterizations of the query lineages they can
represent~\cite{jha2011knowledge}.
Both classes enjoy tractable probabilistic computation: 
for width-based classes, using \emph{message passing}~\cite{lauritzen1988local}, in time linear in the circuit and 
exponential in the treewidth; for OBDDs and d-SDNNFs, in linear time by
definition of the class~\cite{darwiche2001tractable}.
Hence the question that we study: can we compile width-based classes efficiently into
structure-based classes?

We first study how to perform this transformation, and show 
corresponding \emph{upper bounds}.
Existing work has
already studied the compilation of bounded-pathwidth circuits to
OBDDs~\cite{jha2012tractability}, which can be made constructive
\cite[Lemma~6.9]{amarilli2016tractable}.
Accordingly, we focus on compiling
\emph{bounded-treewidth circuits} to \emph{d-SDNNF circuits}.
Our first contribution, stated in Section~\ref{sec:result} and proved in
Section~\ref{sec:proof}, is to show the following:
\begin{result}[(Theorem~\ref{thm:upper_bound} and subsequent remark)]
\label{res:upper}
  Given as input a Boolean circuit $C$ of treewidth $k$,
  we can compute a d-SDNNF equivalent to $C$
  in time $O(|C| \times f(k))$ where $f$ is singly exponential.
\end{result}
The algorithm transforms the input circuit bottom-up,
considering all possible valuations of the gates in
each bag of the tree decomposition, and keeping track of additional information to
remember which guessed values have been substantiated by a corresponding input.
Our result relates to a recent theorem of Bova and
Szeider in~\cite{bova2017circuit}, except that our bound depends on~$\card{C}$
(the circuit size) whereas their bound depends on the number of variables
of~$C$. In exchange for this, we improve on their result in two ways. First,
our result is constructive, whereas~\cite{bova2017circuit} only shows a bound on the
size of the d-SDNNF, without bounding the complexity of effectively computing
it. Second, our bound is singly exponential in~$k$,
whereas~\cite{bova2017circuit} is doubly exponential;
this allows us to 
be competitive with message passing (also singly exponential in~$k$),
and we believe it can be useful for practical applications.
Indeed, beyond 
probabilistic query evaluation,
our result implies that all
tractable tasks on d-SDNNFs (e.g.,
enumeration~\cite{amarilli2017circuit} and MAP
inference~\cite{fierens2015inference})
are also tractable on bounded-treewidth circuits.

Second, we study \emph{lower bounds} on how efficiently we
can convert from width-based to structure-based classes. Our bounds already apply
to a weaker formalism of width-based circuits, namely monotone CNFs and DNFs of
bounded width, so we focus on them. Our second contribution (in
Section~\ref{sec:obddlower}) 
concerns pathwidth and OBDD representations: we show that,
up to factors in the formula arity (maximal size of clauses) and degree (maximal number
of variable occurrences), any OBDD for a monotone CNF or DNF must
be of width exponential in the pathwidth of the formula. Formally:

\begin{result}[(Theorem~\ref{thm:obddlower})]
  \label{res:lower_obdd}
  Let $\phi$ be a monotone DNF or monotone CNF, let $a\colonequals\arity(\phi)$ and
  $d\colonequals\degree(\phi)$. Then any OBDD for $\phi$ has width
  $2^{\Omega\left(\frac{\pw(\phi)}{a^3 \times d^2}\right)}$.
\end{result}
This result generalizes several existing lower bounds in knowledge
compilation that 
exponentially separate CNFs from OBDDs,
such as~\cite{devadas93comparing} and~\cite[Theorem~19]{bova2017compiling}. 

Our third contribution (Section~\ref{sec:sddnnflower}) is to show an analogue 
for treewidth and (d-)SDNNFs:
\begin{result}[(Theorem~\ref{thm:dSDNNFlower})]
  \label{res:lower_dsdnnf}
  Let $\phi$ be a monotone DNF (resp., monotone CNF), let $a\colonequals\arity(\phi)$ and
  $d\colonequals\degree(\phi)$. Then any d-SDNNF (resp., SDNNF) for $\phi$ has size
  $2^{\Omega\left(\frac{\tw(\phi)}{a^3 \times d^2}\right)}$.
\end{result}

Our two lower bounds contribute to a vast landscape of knowledge
compilation results
giving lower bounds on compiling 
specific Boolean functions to
restricted
circuits classes, e.g.,
\cite{devadas93comparing,razgon2014obdds,bova2017compiling} to OBDDs,
\cite{cali2017non} to \emph{decision} structured DNNF,
\cite{beame2015new} to \emph{sentential decision diagrams} (SDDs),
\cite{pipatsrisawat2010lower,bova2016knowledge} to d-SDNNF,
\cite{bova2016knowledge,capelli2016structural,capelli2017understanding} to d-DNNFs and DNNFs.
However, all those lower bounds (with the exception of some results
in~\cite{capelli2016structural,capelli2017understanding} discussed
in Section~\ref{sec:sddnnflower}) apply to
well-chosen families of Boolean functions (usually CNF),
whereas Result~\ref{res:lower_obdd} and~\ref{res:lower_dsdnnf} 
apply to \emph{any} monotone CNF and DNF. Together with
Result~\ref{res:upper},
these generic lower bounds point to a strong relationship
between width parameters and structure representations, on monotone CNFs and
DNFs of constant arity and degree. Specifically, the smallest width of OBDD
representations of any such formula $\phi$ is in~$2^{\Theta(\pw(\phi))}$, i.e.,
precisely singly exponential in the pathwidth; and
an analogous bound applies to d-SDNNF size and treewidth of DNFs.

To prove our lower bounds, we
rephrase pathwidth and treewidth to new notions of
\emph{pathsplitwidth} and \emph{treesplitwidth}, 
which
intuitively measure the performance of a variable ordering or v-tree. 
We also use the \emph{disjoint non-covering prime implicant sets} (dncpi-sets),
a tool introduced in~\cite{amarilli2016tractable,amarilli2016leveraging} by some of the present
authors, and generalizing \emph{subfunction width}~\cite{bova2017compiling}.
These dncpi-sets allow us to derive lower bounds on OBDD width directly
using~\cite{amarilli2016leveraging}. We show how they can also
imply lower bounds on d-SDNNF size, using the recent communication complexity approach
of Bova, Capelli, Mengel and Slivovsky~\cite{bova2016knowledge}.

Our fourth contribution (Section~\ref{sec:lineages}) 
applies our lower bounds
to intensional query evaluation on relational databases.
We reuse the notion of \emph{intricate} queries 
of~\cite{amarilli2016tractable},
and show that d-SDNNF representations of the lineage of these queries
have size exponential in the treewidth of \emph{any} input instance.
This extends the result of~\cite{amarilli2016tractable} from OBDDs to d-SDNNFs:
\begin{result}[(Theorem~\ref{thm:querymain})]
  \label{res:lineages}
  There is a constant $d\in\NN$ such that the following is true.
  Let $\sigma$ be an arity-2 signature, and $Q$ be a connected $\ucqneq$ which
  is intricate on~$\sigma$. For any instance $I$ on~$\sigma$, any d-SDNNF representing the
  lineage of~$Q$ on~$I$ has size $\geq 2^{\Omega(\tw(I)^{1/d})}$.
\end{result}

As in~\cite{amarilli2016tractable}, this result shows that, on arity-2
signatures and under constructibility assumptions, treewidth is the right
parameter on instance families to ensure that all queries (in monadic
second-order) have tractable d-SDNNF lineage representations.

We start in Section~\ref{sec:preliminaries} with preliminaries.
Full proofs of all results are 
in the appendix.

This paper is an extended version of the conference publication~\cite{amarilli2018connecting}.

\section{Preliminaries}
\label{sec:preliminaries}
\myparagraph{Hypergraphs, treewidth, pathwidth}
A \emph{hypergraph} $H=(V,E)$ consists of a finite set of \emph{nodes}
(or \emph{vertices}) $V$ and of a set $E$ of \emph{hyperedges} (or simply
\emph{edges}) which are non-empty subsets of~$V$.
We always assume that hypergraphs have at least one edge.
For a node~$v$ of~$H$, we write $E(v)$
for the set of edges of~$H$ that contain $v$.
The \emph{arity} of~$H$, written
$\arity(H)$,
is the maximal size of an edge of~$H$.
The \emph{degree} of~$H$,
written $\degree(H)$,
is the maximal number of edges to which a
vertex belongs, i.e., $\max_{v\in V} \card{E(v)}$.

A \emph{tree decomposition} of a hypergraph $H = (V, E)$
is a finite, rooted tree $T$, whose nodes~$b$ (called
\emph{bags}) are labeled by a subset $\dom(b)$ of~$V$,
and which satisfies:
\begin{enumerate}
  \item for every
fact $e \in E$, there is a bag $b \in T$ with $e \subseteq
\dom(b)$; 
\item for all $v \in V$, the set of bags $\{b \in T \mid v \in
  \dom(b)\}$ is a connected subtree of~$T$.
\end{enumerate}
For brevity, we identify a bag $b$ with its domain $\dom(b)$.
The \emph{width} of~$T$ is $\max_{b\in T} \card{\dom(b)}-1$.
The \emph{treewidth} of~$H$ is the minimal width of a tree decomposition of~$H$.
Pathwidth is defined similarly 
but with \emph{path decompositions},
where $T$ is a path rather than a
tree.

It is NP-hard to determine the treewidth of a hypergraph,
but 
we can compute a tree decomposition in linear time
when parametrizing by the treewidth:
\begin{theorem}[\cite{bodlaender1996linear}]
	\label{thm:bod}
	Given a hypergraph $H$ and an integer $k \in \NN$ we can check in time
	$O(|H| \times g(k))$ whether $H$ has treewidth $\leq k$, and if yes
        output a tree decomposition of $H$ of width $\leq k$, where $g$ is a
        fixed function
	in $O(2^{(32 + \epsilon)k^3})$ for any $\epsilon > 0$.
\end{theorem}

For simplicity, we will often assume that a tree decomposition is \emph{nice},
meaning that:
        \begin{minipage}{\linewidth}
\medskip %
	\begin{enumerate}
		\item 
                it is a full binary tree, i.e., each node has exactly zero or two children;
              \item
                for every internal bag $b$ with children $b_l,b_r$ we have $b \subseteq b_l \cup b_r$;
              \item
                for every leaf bag $b$ we have $|b| \leq 1$;
              \item 
                for every non-root bag $b$ with parent $b'$, we have
                  $\card{b \setminus b'} \leq 1$;
                \item
                  for the root bag $b$ we have $\card{b} \leq 1$.
        \end{enumerate}
        \end{minipage}
\begin{lemmarep}
	\label{lem:nice}
	Given a tree decomposition $T$ of width $k$ having $n$ nodes, we can
        compute in time $O(k\times n)$ a nice tree decomposition $T'$ of width
        $k$ having $O(k \times n)$ nodes.
\end{lemmarep}
\begin{proof}
  We first make the tree decomposition binary (but not necessarily full) by
  replacing each bag $b$ with children $b_1, \ldots, b_n$ with $n > 2$ by a
  chain of bags with the same label as~$b$ to which we attach the children
  $b_1, \ldots, b_n$. This process is in time $O(n)$ and does not change the
  width.

  We then ensure the second and third conditions, by applying a transformation
  to leaf bags and to internal bags. We modify every leaf bag $b$ containing
  more than one vertex by a chain of at most $k$ internal bags with leaves where
  the vertices are added one after the other. At the same time, we modify every
  internal bag $b$ that contains elements $v_1, \ldots, v_n$ not present in the
  union $D$ of its children: we replace $b$ by a chain of at most $k$ internal bags
  $b_1', \ldots, b_n'$ containing respectively $b, b \setminus \{v_n\}, b
  \setminus \{v_n, v_{n-1}\}, \ldots, D$, each bag having a child introducing
  the corresponding gate $v_i$. This is in time $O(kn)$, and again
  it does not change the width; further, the result of the process satisfies the
  second and third conditions and obviously it is still a binary tree.

  We next ensure the fourth and fifth condition. To do this, whenever a non-root bag $b$ and its
  parent $b'$ violate the condition, we create a chain of intermediate nodes
  where the gates in $b \setminus b'$ are removed one after the other; and we
  replace the root bag by a chain of bags where the elements of the root bag are
  removed one after the other. This does
  not affect the fact that the tree decomposition is binary, or the second and
  third conditions, it does not change the width, and it runs in time $O(kn)$.
  Observe that the bound is $O(kn)$ in the original tree decomposition (not in
  the output of the previous step), because the transformation never needs to be
  performed within the chains of nodes that we have introduced in the previous
  step; it only needs to be performed on interfaces between bags that correspond
  to interfaces between original bags in~$T$.

  The only missing part is to ensure that the tree decomposition is full, which
  we can simply ensure in linear time by adding bags with an empty label as a
  second children for internal nodes that have only one child. This is obviously
  in linear time, does not change the width, and does not affect the other
  conditions, concluding the proof.
\end{proof}

\myparagraph{Boolean circuits and functions}
A (Boolean) \emph{valuation} of a set $V$ is a function $\nu: V \to \{0, 1\}$.
A \emph{Boolean function} $\phi$ on variables~$V$ is a mapping that associates
to each valuation $\nu$ of~$V$ a Boolean value in $\{0, 1\}$ called the
\emph{evaluation} of~$\phi$ according to~$\nu$.

A \emph{(Boolean) circuit} $C = (G, W, g_\out, \mu)$ is a directed acyclic graph $(G,
W)$ whose vertices~$G$ are called \emph{gates}, whose edges $W$ are called
\emph{wires},
where $g_\out \in G$
is the \emph{output gate},
and where each gate
$g \in G$ has a \emph{type} $\mu(g)$ among $\var$ (a \emph{variable
gate}), $\NOT$,
$\OR$, $\AND$.
The \emph{inputs}
of a gate $g \in G$ are the gates $g' \in G$
such that $(g', g) \in W$; the \emph{fan-in} of~$g$ is its number of inputs.
We require $\NOT$-gates to have fan-in~1 and
$\var$-gates to have fan-in~0.
The \emph{treewidth} of~$C$, and its \emph{size}, are those of the graph~$(G,
W)$.
The set $C_\var$ of \emph{variable gates} of~$C$ are those of type~$\var$.
Given a valuation $\nu$ of~$C_\var$, we extend it to an
\emph{evaluation} of~$C$ by mapping each variable $g \in C_\var$
to~$\nu(g)$, and evaluating the other gates according to their type.
The Boolean function on~$C_\var$ \emph{captured} by the circuit
is the one that maps~$\nu$ to the evaluation of~$g_\out$ under~$\nu$.
Two circuits are \emph{equivalent} if they capture the same function. 

We recall restricted circuit classes from knowledge compilation.
We say that $C$ is in \emph{negation normal form} (NNF)
if the inputs of $\NOT$-gates are always variable gates.
For a gate~$g$ in a Boolean circuit~$C$,
we write $\VARS(g)$ for the set of variable gates of~$C_\var$
that have a directed path to~$g$ in~$C$.
An $\AND$-gate $g$ of $C$ is \emph{decomposable} if for every two input gates
$g_1\neq g_2$ of $g$ we have $\VARS(g_1) \cap \VARS(g_2) = \emptyset$.
We call $C$ \emph{decomposable} if each $\AND$-gate is.

A stronger requirement than
decomposability is \emph{structuredness}.
  A \emph{v-tree}~\cite{pipatsrisawat2008new} over a set~$V$ is a rooted ordered
  binary tree $T$ whose leaves
        are in bijection with $V$; we identify each leaf with the associated
        element of $V$.
	For $n \in T$, we denote by $T_n$ the subtree of~$T$ rooted at~$n$, and 
	for a subset $U \subseteq T$ of nodes of~$T$,
        we denote by $\LEAVES(U)$ the leaves that are in $U$, i.e., $U \cap V$.
        We say that $T$ \emph{structures} a Boolean circuit $C$ (and call it a
        \emph{v-tree for~$C$}) if $T$ is over the
        set~$C_\var$ and if, for every $\AND$-gate $g$ of~$C$ with inputs $g_1, \ldots,
        g_m$ and $m>0$, there is a node $n \in T$ that \emph{structures} $g$,
        i.e., $n$ has $m$ children $n_1, \ldots, n_m$ and we have
         $\VARS(g_i) \subseteq \LEAVES(T_{n_i})$ for all $1 \leq i \leq m$.
        We call~$C$ \emph{structured} if some v-tree structures it.
Note that structured Boolean circuits are always decomposable,
and their $\AND$-gates have at most two inputs because $T$ is binary.

A last requirement on circuits is \emph{determinism}. An $\OR$-gate $g$ of $C$ is \emph{deterministic} if there
is no pair $g_1\neq g_2$ of input gates of~$g$ and valuation $\nu$ of $C_\var$ such
that 
$g_1$ and $g_2$ both evaluate to~$1$ under~$\nu$.
A Boolean circuit is \emph{deterministic} if each $\OR$-gate is.

The main structured class of circuits that we study in this work are
\emph{deterministic structured decomposable NNFs}, which we denote d-SDNNF for
brevity as in~\cite{pipatsrisawat2008new}.

\myparagraph{DNFs and CNFs}
We also study other representations of Boolean functions, namely,
Boolean formulas in \emph{conjunctive normal form} (\emph{CNFs}) and in
\emph{disjunctive normal form} (\emph{DNFs}). 
A DNF (resp., CNF) $\phi$ on a set of variables $V$ is a disjunction (resp.,
conjunction) of
\emph{clauses},
each of which
is a conjunction (resp., disjunction)  of \emph{literals} on~$V$, i.e.,
variables of~$V$ (a \emph{positive} literal) or their negation
(a \emph{negative} literal).
A \emph{monotone DNF} (resp., monotone CNF) is one where all literals are
positive, in which case
we often identify a clause to the set of variables
that it contains.
We always assume that monotone DNFs and monotone CNFs are
\emph{minimized}, i.e., no clause 
is a subset
of another.
This ensures that every monotone Boolean function has a unique
representation as a monotone DNF (the disjunction of its prime
implicants), and likewise for CNF.
We assume that CNFs and DNFs always contain at least
one non-empty clause (in particular, they cannot represent constant functions).
Monotone DNFs and CNFs~$\phi$ are isomorphic to hypergraphs:
the vertices are the variables of~$\phi$, and the hyperedges are the clauses
of~$\phi$.
We often identify $\phi$ to its hypergraph.
In particular, the \emph{pathwidth}
and \emph{treewidth} of~$\phi$, and its \emph{arity} and \emph{degree},
are defined as that of its hypergraph.

\section{Upper Bounds}
\label{sec:result}
Our upper bound result
studies how to
compile a Boolean circuit to a d-SDNNF, parametrized by the
treewidth of the input circuit. 
To present it, we first review the 
independent result that was recently 
shown by Bova and Szeider~\cite{bova2017circuit} about these circuit classes:

\begin{theorem}[{\cite[Theorem~3 and Equation~(22)]{bova2017circuit}}]
  Given a Boolean circuit $C$ with $n$~variables and of treewidth~$\leq k$, there exists an
  equivalent d-SDNNF of size $O(f(k) \times n)$, where $f$~is doubly
  exponential.
\end{theorem}

An advantage of their result is that
it depends only on the \emph{number of variables} of the circuit (and on the
width parameter), not on the
\emph{size} of the circuit. None of our results will have this advantage, and we will always measure
complexity as a function of the size of the input circuit.
In exchange for this advantage, their result has two drawbacks:
(i) the doubly exponential
dependency on the width; and (ii) its nonconstructive aspect, because
\cite{bova2017circuit} gives no time bound
on the computation, 
leaving open the
question of 
effectively compiling 
bounded-treewidth circuits to d-SDNNFs.

\myparagraph{Naive constructive bound}
We first address the second drawback by showing an easy constructive result. The argument is very
simple and appeals to techniques from our earlier works on provenance
circuits~\cite{amarilli2015provenance,amarilli2016tractable};
it is independent from the techniques of~\cite{bova2017circuit}.

\begin{theoremrep}
  \label{thm:naive}
  Given any circuit $C$ of treewidth~$k$, we can compute an equivalent d-SDNNF
  in linear time parametrized by~$k$, i.e., in time $O(\card{C}\times f(k))$ for
  some computable function~$f$.
\end{theoremrep}

\begin{proofsketch}
  We encode in linear time the input circuit $C$ to a relational instance $I$
  with same treewidth. We use
  \cite[Theorem~6.11]{amarilli2016tractable} to
  construct in linear time a provenance representation $C'$ on~$I$ of a
  fixed MSO formula that describes Boolean circuit evaluation. This
  allows us to obtain in linear time from~$C'$ the desired equivalent d-SDNNF
  representation.
\end{proofsketch}

\begin{proof}
  We first define a signature $\sigma$ to encode circuits into relational
  instances: the elements of such an instance are gates, and $\sigma$ features
  unary relations $\mathrm{Or}$, $\mathrm{And}$, $\mathrm{Not}$,
  $\mathrm{Variable}$ to describe the type of each gate, a unary relation
  $\mathrm{Output}$ to identify the output gate, and a binary relation
  $\mathrm{Wire}$ to describe the wires. We also add a unary relation
  $\mathrm{True}$ to describe a valuation: the relation applies to variable
  gates to indicate whether they are true or not.
  
  Now, we write an MSO formula $\phi$ on~$\sigma$ which checks whether a circuit
  evaluates to true under the indicated valuation. The formula $\phi$
  existentially guesses a set $S$ of true gates: it asserts (i) that the variable
  gates of~$S$ are exactly the ones indicated as $\mathrm{True}$ in the input
  instance, (ii) that the output gate of the circuit is in~$S$, and (iii) that
  $S$ satisfies the semantics of internal gate, i.e., each $\OR$-gate is in~$S$ iff
  it has an input in~$S$, each $\AND$-gate is in~$S$ iff it has all its inputs
  in~$S$, and each $\NOT$-gate is in~$S$ iff its input is not in~$S$.
  
  Now, given an input circuit $C$ with variable gates~$V$,
  we encode it in linear time into a relational structure $I$ on~$\sigma$ in the
  expected way. For any subset $V'$ of~$V$, we let $I_{V'}$ be $I$ where we add
  a fact $\mathrm{True}(g)$ for each $g \in V$.
  In particular, the instance $I_V$
  can be constructed in linear time, and its treewidth is the same as that
  of~$C$. We now use \cite[Theorem~6.11]{amarilli2016tractable} to compute 
  in linear time in~$I_V$ (hence in~$C$)
  a d-SDNNF\footnote{The result only states that it computes a d-DNNF, not a
  d-SDNNF. However, it is immediate from the construction that the circuit
  actually follows a v-tree, which is given by the tree encoding of the input
  instance. This can be checked from the proof (given as that of 
  \cite[Theorem~3.5.8]{amarilli2016leveraging}).} $C'$ capturing the
  \emph{provenance} of~$\phi$ on~$I_V$, i.e., for any Boolean valuation $\nu$
  of~$I_V$, letting $I' \colonequals \{F \in I_V \mid \nu(F) = 1\}$, we have
  $\nu(C') = 1$ iff $I' \models \phi$. We fix to~$1$ the value of the inputs
  of~$C'$ that stands for facts other than $\mathrm{True}$-facts (i.e., those
  that describe the structure of the circuit). The result $C''$ is computed in
  linear time from~$C$, it is still a
  d-SDNNF, and it ensures that for any Boolean valuation $\nu$ of~$V$, letting
  $V' \colonequals \{g \in V \mid \nu(g) = 1\}$, we have $\nu(C') = 1$ iff
  $I_{V'} \models \phi$. By construction, the latter holds iff $\nu(C) = 1$. In
  other words, $C''$ is equivalent to~$C$, which concludes the proof.
\end{proof}

This result shows that we can effectively compile in
linear time parametrized by the treewidth~$k$, but does not address the first drawback, namely, the dependency in~$k$.

\myparagraph{Improved bound}
Our main upper bound result subsumes the naive bound above, with a more
elaborate proof, again independent of the techniques
of~\cite{bova2017circuit}. It addresses both drawbacks and shows that we can effectively compile
in time singly exponential in~$k$; formally:

\begin{theorem}
\label{thm:upper_bound}
  Given as input a Boolean circuit $C$ and tree decomposition $T$ of width $k$,
  we can compute a d-SDNNF equivalent to $C$ with its v-tree,
  in $O\left(|T| \times 2^{(4+\epsilon) k}\right)$ for any~$\epsilon > 0$.
\end{theorem}

We prove Theorem~\ref{thm:upper_bound}
in the next section.
Observe how we assume the tree decomposition to be given as part of
the input. If it is not, we can compute one with Theorem~\ref{thm:bod},
but this becomes the bottleneck: the complexity becomes
$O\left(\card{C} \times 2^{(32 + \epsilon)k^3}\right)$ for any~$\epsilon>0$. 

\myparagraph{Applications}
Theorem~\ref{thm:upper_bound} implies several consequences for bounded-treewidth
circuits. The first one deals with \emph{probability computation}: we are given
a \emph{probability valuation}~$\pi$ mapping each variable $g \in C_\var$ to a 
probability that~$g$ is true (independently from other variables),
and we wish to compute the probability
$\pi(C)$ 
that $C$ evaluates to true under~$\pi$, assuming that arithmetic operations (sum
and product) take unit time.
This problem is \#P-hard for arbitrary
circuits, but it is tractable for d-SDNNF \cite{darwiche2001tractable}. Hence,
our result implies the following, where $\card{\pi}$
denotes the size of writing the probability valuation~$\pi$:

\begin{corollaryrep}
  \label{cor:proba}
  Given a Boolean circuit $C$, a tree decomposition $T$ of width~$k$ of~$C$, and
  a probability valuation $\pi$ of~$C$, we can compute~$\pi(C)$ in
  $O\left(\card{\pi} + \card{T} \times 2^{(4+\epsilon) k}\right)$ for any $\epsilon > 0$.
\end{corollaryrep}

\begin{proof}
  Use Theorem~\ref{thm:upper_bound} to compute an equivalent d-SDNNF $C'$; as
  $C$ and $C'$ are equivalent, it is clear that $\pi(C) = \pi(C')$. Now, compute
  the probability $\pi(C')$ in linear time in~$C'$ and~$\card{\pi}$ by a simple
  bottom-up pass, using the fact that $C'$ is a d-DNNF
  \cite{darwiche2001tractable}.
\end{proof}

This improves the bound obtained when applying message passing
techniques~\cite{lauritzen1988local} directly on the bounded-treewidth input
circuit (as presented, e.g., in
\cite[Theorem~D.2]{amarilli2015provenance_extended}). Indeed, message passing
applies to \emph{moralized} representations of the input: for each gate,
the tree decomposition must contain a bag containing all inputs of this gate
\emph{simultaneously}, which is problematic for circuits of large fan-in. 
Indeed, if the original circuit has a tree decomposition of width~$k$,
rewriting it to make it moralized 
results in a tree decomposition of width~$3k^2$ (see~\cite[Lemmas~53
and~55]{amarilli2017combined_long}), and
the bound of
\cite[Theorem~D.2]{amarilli2015provenance_extended} then yields an overall complexity 
of~$O\big(|\pi|+|T|\times 2^{3k^2}\big)$ for message passing.
Our Corollary~\ref{cor:proba}
achieves a more favorable bound because Theorem~\ref{thm:upper_bound}
uses directly the associativity of $\AND$ and $\OR$.
We note that the connection between message-passing
techniques and structured circuits has also been investigated by
Darwiche,
but his result
\cite[Theorem~6]{darwiche2003differential} produces arithmetic circuits rather
than d-DNNFs, and it also needs the input to be moralized.

A second consequence concerns the task of \emph{enumerating} the accepting
valuations of circuits, i.e., producing them one after the other, with small
\emph{delay} between each accepting valuation. 
The valuations are concisely represented as \emph{assignments}, i.e., as a
set of variables that are set to true, omitting those that are set to false.
This task is of course NP-hard on
arbitrary circuits (as it implies that we can check whether an accepting
valuation exists), but was recently shown in~\cite{amarilli2017circuit}
to be feasible on d-SDNNFs with
linear-time preprocessing and delay linear in the Hamming weight of each
produced assignment. Hence, we have:

\begin{corollaryrep}
  \label{cor:enum}
  Given a Boolean circuit $C$ and a tree decomposition $T$ of width~$k$ of~$C$,
  we can enumerate the accepting assignments of~$C$ with preprocessing in $O\left(\card{T}
  \times 2^{(4+\epsilon) k}\right)$ and delay linear in the size of each produced
  assignment.
\end{corollaryrep}

\begin{proof}
  Use Theorem~\ref{thm:upper_bound} to compute an equivalent d-SDNNF $C'$, which
  has the same accepting valuations, along with a v-tree $T'$ of~$C'$. We now
  conclude using \cite[Theorem~2.1]{amarilli2017circuit}.
\end{proof}

Other applications of Theorem~\ref{thm:upper_bound} include counting the
number of satisfying valuations of the circuit (a
special case of probability computation), MAP inference
\cite{fierens2015inference} or random sampling of possible worlds (which
can be done on the d-SDNNF in an easy manner).

\section{Proof of the Main Upper Bound Result}
\label{sec:proof}
In this section, we present the construction used to prove
Theorem~\ref{thm:upper_bound}.
We start with prerequisites,
and then describe how to build the d-SDNNF equivalent to the input
bounded-treewidth circuit. Last, we sketch the correctness proof.

\myparagraph{Prerequisites}
Let $C$ be the input circuit, and $T$ the input tree decomposition.
By Lemma~\ref{lem:nice}, we assume that $T$ is nice.
Further, up to adding a constant number of bags and re-rooting~$T$, we can
assume that the root bag of~$T$ contains only the output gate $g_\out$.
For any bag~$b$ of~$T$, we define $\IT(b)$ to be the set of variable gates such
that $b$ is the topmost bag in which they appear; as $T$ is nice, $\IT(b)$ is
either empty or is a singleton $\{g\}$, in which case we call~$b$
\emph{responsible for the variable gate~$g$}.
We can explicitly compute the function $\IT$ in~$O(|T|)$,
i.e.,
compute $\IT(b)$ for each $b \in T$; see for instance
\cite[Lemma~3.1]{flum2002query}.

To abstract away the type of gates and their values in the construction,
we will talk of \emph{strong} and \emph{weak} values. Intuitively, a value is
\emph{strong} for a gate $g$ if any input $g'$ of~$g$ which carries this value
determines the value of~$g$; and \emph{weak} otherwise. Formally:
 
\begin{definition}
  Let $g$ be a gate and $c \in \{0, 1\}$:
  \begin{itemize}
    \item If $g$ is an $\AND$-gate, we say that $c=0$ is \emph{strong} for~$g$ and $c=1$
      is \emph{weak} for~$g$;
    \item If $g$ is an $\OR$-gate, we say that $c=1$ is \emph{strong} for~$g$ and $c=0$
      is \emph{weak} for~$g$;
    \item If $g$ is a $\NOT$-gate, $c=0$ and $c=1$ are both \emph{strong}
      for~$g$;
    \item For technical convenience, if $g$ is a $\var$-gate, $c=0$ and $c=1$ are both \emph{weak} for~$g$.
  \end{itemize}
\end{definition}

If we take any valuation $\nu:C_\var\to\{0,1\}$ of the circuit $C = (G, W,
g_\out,
\mu)$,
and extend it to an evaluation $\nu:G\to\{0,1\}$, then $\nu$ will respect the
semantics of gates. In particular, it will \emph{respect strong values}: 
for any gate $g$ of~$C$, if $g$ has an input $g'$ for which $\nu(g')$ is a strong
value,
then $\nu(g)$ is determined by~$\nu(g')$, specifically, it is~$\nu(g')$ if~$g$ is
an $\OR$- or an $\AND$-gate, and $1-\nu(g')$ if $g$ is a $\NOT$-gate. In our
construction, we will need to guess how gates of the circuit are evaluated,
focusing on a subset of the gates (as given by a bag of~$T$); we will then call
\emph{almost-evaluation} an assignment of these gates that respects strong
values. Formally:
\begin{definition}
	Let $U$ be a set of gates of~$C$.
        We call $\nu: U \to \{0, 1\}$ a \emph{$(C,U)$-almost-evaluation} if 
        it \emph{respects strong values}, i.e., for every gate $g \in U$, if
        there is an input~$g'$ of~$g$ in~$U$ and $\nu(g')$ is a strong value for
        $g$, then $\nu(g)$ is determined from $\nu(g')$ in the sense above.
\end{definition}

Respecting strong values is a necessary condition for such an assignment to be extensible to a
valuation of the entire circuit. However, it is not sufficient: an
almost-evaluation $\nu$ may map
a gate $g$ to a strong value even though $g$ has no input that can justify this
value. This is hard to avoid: when we focus on the
set~$U$, we do not know about other inputs of~$g$. For now, let us call
\emph{unjustified} 
the gates of~$U$ that carry a strong value that is not justified by~$\nu$:

\begin{definition}
	Let $U$ be a set of gates of a circuit $C$ and $\nu$ a $(C,U)$-almost-evaluation.
        We call $g \in U$ \emph{unjustified} if $\nu(g)$ is a strong value
        for~$g$, but, for every input $g'$ of~$g$ in~$U$, the value $\nu(g')$ is
        weak for~$g$; otherwise, $g$ is \emph{justified}.
        The set of unjustified gates is written~$\UNF(\nu)$.
\end{definition}

Let us start to explain
how to construct the d-SDNNF circuit $D$ equivalent to the input
circuit~$C$. We do so by traversing
$T$ bottom-up,
and for each bag $b$ of $T$ we create gates~$G_b^{\nu,S}$ in~$D$,
where $\nu$~is a $(C,b)$-almost-evaluation
and $S$ is a subset of~$\UNF(\nu)$ which we call the \emph{suspicious gates}
of~$G_b^{\nu,s}$. We will connect the gates of~$D$
created for each internal bag~$b$ with the
gates created for its children in~$T$, in a 
way that we will 
explain later.
Intuitively, for a gate $G_b^{\nu,S}$ of~$D$,
the \emph{suspicious gates} $g$ in the set~$S$ are gates of~$b$
whose strong value is not justified by~$\nu$ (i.e., $g \in \UNF(\nu)$),
and is not justified either by any of the almost-evaluations at descendant bags
of~$b$ to which $G_b^{\nu,S}$ is connected.
We call \emph{innocent} the other gates of~$b$; they are the gates that are 
justified in~$\nu$ (in particular, those who carry weak values),
and the gates that are unjustified in~$\nu$ but have been
justified by an almost-evaluation at a descendant bag $b'$ of~$b$. Crucially, in
the latter case, the gate $g'$ justifying the strong value in~$b'$ may no longer
appear in~$b$, making $g$ unjustified for~$\nu$; this is why we remember the set~$S$.

We still have to explain 
how we connect the gates $G_b^{\nu,S}$ of~$D$ to the gates
$G_{b_l}^{\nu_l,S_l}$ and~$G_{b_r}^{\nu_r,S_r}$ created for the children $b_l$
and $b_r$ of~$b$ in~$T$.
The first condition is that $\nu_l$ and~$\nu_r$ must \emph{mutually agree},
i.e., $\nu_l(g) = \nu_r(g)$ for all $g \in b_l \cap b_r$, and $\nu$ must then be
the union of~$\nu_l$ and~$\nu_r$, restricted to~$b$. Remember that~$T$ is
nice, so $b$ is a subset of $b_l \cup b_r$, and it is easy to verify that $\nu$
is then a $(C,b)$-almost-evaluation.
We impose a second condition to prohibit suspicious gates from escaping
before they have been justified, which we formalize as \emph{connectibility} of
a pair $(\nu,S)$ at bag~$b$ to the parent bag of~$b$.
\begin{definition}
	\label{def:connectible}
	Let $b$ be a non-root bag, $b'$ its parent bag,
        and $\nu$ a $(C,b)$-almost-evaluation.
        For any set $S \subseteq \UNF(\nu)$,
        we say that $(\nu,S)$ is \emph{connectible} to~$b'$
        if $S \subseteq b'$, i.e., the suspicious gates of $\nu$ must still
        appear in~$b'$.
\end{definition}
If a gate $G^{\nu,S}_b$ is such that $(\nu,S)$ is not connectible to the parent
bag~$b'$, then this gate will not be used as input to any other gate (but we do
not try to preemptively remove these useless gates in the construction).
We are now ready to give the formal definition that will be used to
explain how gates
are connected:
\begin{definition}
	\label{def:result}
	Let $b$ be an internal bag with children $b_l$ and $b_r$, 
        let $\nu_l$ and
        $\nu_r$ be respectively $(C,b_l)$ and $(C,b_r)$-almost-evaluations that
        mutually agree,
        and $S_l \subseteq \UNF(\nu_l)$
        and $S_r \subseteq \UNF(\nu_r)$ be sets of
	suspicious gates 
        such that both $(\nu_l,S_l)$ and $(\nu_r,S_r)$ are connectible to~$b$.
        The \emph{result} of $(\nu_l,S_l)$ and $(\nu_r,S_r)$ is then defined as
        the pair $(\nu,S)$ where:
	\begin{itemize}
		\item $\nu$ is a $(C,b)$-almost-evaluation defined as the
                  restriction of~$\nu_l \cup \nu_r$ to~$b$.
                \item $S \subseteq \UNF(\nu)$ is the new set of suspicious
                  gates, defined as follows.
                  A gate $g \in b$ is innocent (i.e., $g \in b \setminus S$) if it is
                  justified for $\nu$ or if it is innocent for some child.
			Formally, $b \setminus S \colonequals (b
                        \setminus\UNF(\nu)) \cup \big[b \cap \left[ (b_l \setminus S_l) \cup (b_r
                        \setminus S_r)\right]\big]$.
	\end{itemize}
\end{definition}

\myparagraph{Construction}
We now use these definitions to present the construction formally.
For every variable gate $g$ of $C$,
we create a corresponding variable gate $G^{g,1}$ of $D$,
and we create $G^{g,0} \colonequals \NOT(G^{g,1})$.
For every internal bag $b$ of $T$,
for each $(C,b)$-almost-evaluation $\nu$ and set
$S\subseteq \UNF(\nu)$ of suspicious
gates of $\nu$,
we create 
one $\AND$-gate $G_b^{\nu,S}$
and one $\OR$-gate $G_{b,\children}^{\nu,S}$ which is an input of~$G_b^{\nu,S}$.
For every leaf bag $b$ of $T$, we create 
one gate $G_b^{\nu,S}$ for every $(C,b)$-almost-evaluation $\nu$,
where we set $S \colonequals \UNF(\nu)$; intuitively, in a leaf bag, an unjustified
gate is always suspicious (it cannot have been justified at a descendant bag).

Now, for each internal bag $b$ of~$T$ with children $b_l,b_r$,
for each pair of gates $G_{b_l}^{\nu_l,S_l}$ and~$G_{b_r}^{\nu_r,S_r}$ that are both
connectible to~$b$ and where $\nu_l$ and $\nu_r$ mutually agree, letting $(\nu,
S)$ be the result of~$(\nu_l,S_l)$ and~$(\nu_r,S_r)$, we create 
a gate $G_b^{\nu_l,S_l,\nu_r,S_r}=\AND(G_{b_l}^{\nu_l,S_l}, G_{b_r}^{\nu_r,S_r})$
and make it an input of~$G^{\nu,S}_{b, \children}$.
Last, for each bag $b$ 
which is responsible for a variable gate~$g$ (i.e., $\IT(b) = \{g\}$),
for each $(C,b)$-almost-evaluation $\nu$ and set of suspicious
gates $S \subseteq \UNF(\nu)$,
we set the gate $G^{g,\nu(g)}$ 
to be the second input of $G_b^{\nu,S}$.
The output gate of~$D$ is the gate $G^{\nu,\emptyset}_{b_\root}$ where
$b_\root$ is the root of~$T$
and $\nu$ maps~$g_\out$ to~$1$
(remember that $b_\root$ contains only $g_\out$).

\myparagraph{Correctness}
We have formally described the construction of our d-SDNNF $D$. The construction
clearly works in linear time, and we can prove that the dependency on~$k$ of the
running time is as stated. Further, we easily see that $D$ is structured by a v-tree constructed
from the tree decomposition~$T$. To show that $D$ is equivalent to~$C$, one
direction is easier: any valuation~$\chi$ that satisfies~$C$ also satisfies~$D$,
because we can construct an \emph{accepting trace} in~$D$ using the gates~$G^{\nu,S}_b$
for $\nu$ the restriction of the evaluation~$\chi$ to~$b$, and
for~$S\colonequals \UNF(\restr{\chi}{T_b})$ where~$T_b$ denotes the gates of~$C$ occurring in
the bags of the subtree of~$T$ rooted at~$b$. The converse is trickier: we show
that any accepting trace of~$D$ describes an evaluation of~$C$ that respects strong
values by definition of almost-evaluations, and eventually justifies every
gate which is given a strong value thanks to our bookkeeping of suspicious
gates. Last, we show that $D$ is deterministic: this is unexpected
because we freely guess the values of gates of~$C$ at leaf bags, but it holds because,
when we know the valuation of the variable gates, knowing the valuation of all
gates of a
bag~$b$
uniquely fixes the valuation at the subtree rooted at~$b$.
This concludes the proof sketch of Theorem~\ref{thm:upper_bound}; see
Appendix~\ref{apx:correctness_proof} for the full proof.

\begin{toappendix}
\label{apx:correctness_proof}

We now prove that the circuit $D$ constructed in 
  the main text is indeed a d-SDNNF equivalent to the initial circuit
  $C$, and that it can be constructed together with its v-tree in $O\left(|T|\times
  2^{(4+\epsilon)k}\right)$ for any~$\epsilon>0$.

\subsection{$D$ is a Structured DNNF}
Negations only apply to the input gates, so $D$ is an NNF.
To justify that $D$ is structured, we will define an appropriate v-tree $T'$. 
Consider the nice tree decomposition $T$ of $C$ that was used to construct $D$.
For each bag $b \in T$, the v-tree $T'$ has a node $b'$.
For each internal bag $b \in T$ with children $b_1,b_2$, $T'$ has a node $b'_{\children}$, whose children are $b'_1$ and $b'_2$, and whose
parent is $b'$. For every bag $b \in T$ that is responsible of some variable gate $g$, $T'$ has a node $G^{g,1}$ whose parent is $b'$. 
Hence $T'$ is a tree, and one can check that any
any $G^{\nu_l,I_l,\nu_r,I_r}_b$ is structured by $b'_{\children}$ and any $G_b^{\nu,S}$ is structured by $b'$, so that $T'$ structures $D$.

\subsection{$D$ is Equivalent to $C$}
In order to prove that $D$ is equivalent to $C$, we introduce the
  standard notion of a \emph{trace} in an NNF:

\begin{definition}
	Let $D$ be an NNF, $\chi$ a valuation of its variable gates, and $g$ a gate that evaluates to~$1$ under $\chi$.
	A \emph{trace} of $D$ starting at $g$ according to $\chi$ is a set $\Xi$ of gates of~$D$ that is minimal by inclusion and such that:
	\begin{itemize}
		\item $g \in \Xi$;
		\item If $g' \in \Xi$ and $g'$ is an $\AND$ gate, then $W(g') \subseteq \Xi$, where $W(g)$ denotes the set of gates that are input to $g$;
		\item If $g' \in \Xi$ and $g'$ is an $\OR$ gate, then exactly one input of $g'$ that evaluates to~$1$ is in $\Xi$.
	\end{itemize}
\end{definition}

The first step is then to prove that traces have exactly one almost-evaluation
  corresponding to each descendant bag, and that these almost-evaluations mutually agree. 
\begin{lemma}
	\label{lem:only_one_agree2}
	Let $\chi$ be a valuation of the variable gates, $G^{\nu,S}_b$ a gate in $D$ that evaluates to~$1$ under $\chi$ and $\Xi$ be a 
	trace of $D$ starting at $G^{\nu,S}_b$ according to $\chi$.
	Then for any bag $b' \leq b$
        (meaning that $b'$ is $b$ or a descendant of $b$),
        $\Xi$ contains exactly one gate of the form~$G^{\nu',S'}_{b'}$.
	Moreover, over all $b' \leq b$, all the almost-evaluations of the gates $G^{\nu',S'}_{b'}$ that are in~$\Xi$ mutually agree.
\end{lemma}

\begin{proof}
	The fact that $\Xi$ contains exactly one gate $G^{\nu',S'}_{b'}$
        for any bag $b' \leq b$ is obvious by construction of $D$, as
        $\OR$-gates are assumed to have exactly one input evaluated to~1
        in~$\Xi$.
	For the second claim, suppose by contradiction that not all the almost-evaluations of the gates $G^{\nu',S'}_{b'}$ that are in $\Xi$ mutually agree.
	We would then have $G^{\nu_1,S_1}_{b_1}$ and $G^{\nu_2,S_2}_{b_2}$ in $\Xi$ and $g \in b_1 \cap b_2$ such that $\nu_1(g) \neq \nu_2(g)$. 
	But because $T$ is a tree decomposition,
	$g$ appears in all the bags on the path from $b_1$ and $b_2$, and by construction the almost-evaluations of 
	the $G^{\nu',S'}_{b'}$ on this path that are in $\Xi$ mutually agree, hence a contradiction.
\end{proof}

Therefore, Lemma~\ref{lem:only_one_agree2} allows us to define the union of the almost-evaluations in such a trace:
\begin{definition}
	\label{def:gamma2}
	Let $\chi$ be a valuation of the variable gates, $G^\nu_b$ a gate in $D$ that evaluates to~$1$ under $\chi$ 
	and $\Xi$ be a trace of $D$ starting at $G^\nu_b$ according to $\chi$.
	Then $\gamma(\Xi) \defeq \bigcup_{G^{\nu',S'}_{b'} \in \Xi} \nu'$ (the union of the almost-evaluations in $\Xi$, which is a valuation 
	from $\bigcup_{G^{\nu',S'}_{b'} \in \Xi} b'$ 
	to $\{0,1\}$) is properly defined.
\end{definition}

We now need to prove a few lemmas about the behavior of gates that are innocent (i.e., not suspicious).
\begin{lemma}
\label{lem:behaviour}
	Let $\chi$ be a valuation of the variable gates, $G^{\nu,S}_b$ a gate in $D$ that evaluates to~$1$ under $\chi$ and $\Xi$ be a 
	trace of $D$ starting at $G^{\nu,S}_b$ according to $\chi$.
	Let $g \in b$ be a gate that is innocent ($g \notin S$). Then the following holds:
	\begin{itemize}
		\item If $\nu(g)$ is a weak value of $g$, then for every input $g'$ of $g$ that is in the domain of $\gamma(\Xi)$ (i.e., $g'$ appears in a bag $b' \leq b$),
			we have that $\gamma(\Xi)$ maps $g'$ to a weak value of $g$;
		\item If $\nu(g)$ is a strong value of $g$, then there exists an input $g'$ of $g$ that is in the domain of $\gamma(\Xi)$ such that 
			$\gamma(\Xi)(g')$ is $\nu(g)$ if $g$ is an $\AND$
                        or $\OR$ gate, and $\gamma(\Xi)(g')$ is $1-\nu(g)$
                        if $g$ is a $\NOT$ gate.
	\end{itemize}
\end{lemma}
\begin{proof}
	We prove the claim by bottom-up induction on $b \in T$. One can easily check that the claim is true when $b$ is a leaf bag, remembering that
	in this case we must (crucially) have $S = \UNF(\nu)$ by construction (that is, all the gates that are unjustified are suspicious).
	For the induction case, let $b_l$, $b_r$ be the children of $b$.
	Suppose first that $\nu(g)$ is the weak value of $g$, and suppose for a contradiction that there is an input $g'$ of $g$ in the domain of $\gamma(\Xi)$ such that
	$\gamma(\Xi)(g')$ is a strong value of $g$. By the occurrence and connectedness properties of tree decompositions, there exists a bag $b' \leq b$ in which
	both $g$ and $g'$ occur. Consider the  gate $G^{\nu',S'}_{b'}$ that is in $\Xi$:
        by Lemma~\ref{lem:only_one_agree2}, this gate exists and is unique.
	By definition of $\gamma(\Xi)$ we have $\nu'(g')=\gamma(\Xi)(g')$.
	Because $\nu'$ is a $(C,b')$-almost-evaluation that maps $g'$ to a strong value of $g$, we must have that $\nu'(g)$ is also a strong value of $g$,
	thus contradicting our hypothesis that $\nu(g)=\gamma(\Xi)(g) = \nu'(g)$ is a weak value for $g$.

	Suppose now that $\nu(g)$ is a strong value of $g$. We only treat
        the case when $g$ is an $\OR$ or an $\AND$ gate, as the case of a
        $\NOT$ gate is similar.
	We distinguish two sub-cases:
	\begin{itemize}
		\item $g$ is justified. Then clearly, because $\nu$ is a $(C,b)$-almost-evaluation, there must exist an input $g'$ of $g$ that is 
			also in $b$ such that $\nu(g')$ is a strong value of $g$, which proves the claim.
		\item $g$ is unjustified but innocent ($g \notin S$).
			By construction (precisely, by the second item of Definition~\ref{def:result}), $g$ must then be innocent for a child of $b$, and the claim
			clearly holds by induction hypothesis. \qedhere
	\end{itemize}
\end{proof}

Lemma~\ref{lem:behaviour} allows us to show that for a gate $g$, letting $b$ be the topmost bag in which $g$ appears (hence, each input of $g$ must occur in some bag 
$b' \leq b$), if $g$ is innocent then for any trace $\Xi$ starting at a gate for bag $b$, $\gamma(\Xi)$ respects the semantics of $g$.
Formally:

\begin{lemma}
\label{lem:respects_semantics}
	Let $\chi$ be a valuation of the variable gates, $G^{\nu,S}_b$ a gate in $D$ that evaluates to~$1$ under $\chi$ and $\Xi$ be a 
	trace of $D$ starting at $G^{\nu,S}_b$ according to $\chi$.
	Let $g \in b$ be a gate such that $b$ is the topmost bag in which $g$ appears (hence $W(g) \subseteq \text{domain}(\gamma(\Xi))$).
	If $g$  is innocent ($g \notin S$) then $\gamma(\Xi)$ respects the semantics of $g$, that is 
	$\gamma(\Xi)(g)=\bigodot \gamma(\Xi)(W(g))$ where $\bigodot$ is the type of $g$.
\end{lemma}
\begin{proof}
	Clearly implied by Lemma~\ref{lem:behaviour}.
\end{proof}

We need one last lemma about the behavior of suspicious gates, which intuitively tells us that if we have already seen all the input gates of a gate $g$
and $g$ is still suspicious, then $g$ can never escape:
\begin{lemma}
\label{lem:suspicious_not_innocent_propagate}
	Let $\chi$ be a valuation of the variable gates, $G^{\nu,S}_b$ a gate in $D$ that evaluates to~$1$ under $\chi$ and $\Xi$ be a 
	trace of $D$ starting at $G^{\nu,S}_b$ according to $\chi$.
	Let $g$ be a gate such that the topmost bag $b'$ in which $g$ appears is $\leq b$, and consider the unique gate of the form
	$G^{\nu',S'}_{b'}$ that is in $\Xi$.
	If $g \in S'$ then $b'=b$ (hence $G^{\nu,S}_b = G^{\nu',S'}_{b'}$
        by uniqueness).
\end{lemma}
\begin{proof}
Let $g \in S'$. Suppose by contradiction that $b' \neq b$.
Let $p$ be the parent of $b'$ (which exists because $b' < b$).
It is clear that by construction $(\nu',S')$ is connectible to $p$ (recall Definition~\ref{def:connectible}), hence $g$ must be in $p$,
contradicting the fact that $b'$ should have been the topmost bag in which $g$ occurs.
Hence $b'=b$.
\end{proof}

We now have all the results that we need to show that $D \implies C$, i.e., for every valuation~$\chi$ of the variables of~$C$, if $\chi(D) = 1$ then $\chi(C) = 1$. We prove a stronger result:
\begin{lemma}
	\label{lem:corresponds}
	Let $\chi$ be a valuation of the variable gates, $G^{\nu,\emptyset}_{\root(T)} \in D$ a gate that evaluates to~$1$ under $\chi$,
	and $\Xi$ a trace 
	of $D$ starting at $G^{\nu,\emptyset}_{\root(T)}$ according to~$\chi$. Then $\gamma(\Xi)$ corresponds to the evaluation $\chi$ of $C$.
\end{lemma}
\begin{proof}
  We prove by induction on $C$ (as its graph is a DAG) that for all $g \in C$, $\gamma(\Xi)(g) = \chi(g)$. When $g$ is a variable gate, 
consider the bag $b'$ that is responsible of $g$,
  and consider
 the gate~$G^{\nu',S'}_{b'}$ that is in $\Xi$: this gate exists and is unique according to Lemma~\ref{lem:only_one_agree2}. This gate evaluates to $1$ under $\chi$ (because it is in the trace),
 which is only possible if $G^{g,\nu'(g)}$ evaluates to $1$ under $\chi$, hence by construction we must have $\nu'(g)=\chi(g)$ and then $\gamma(\Xi)(g)=\chi(g)$.
 Now suppose that $g$ is an internal gate, and consider the topmost bag $b'$ in which $g$ appears. 
 Consider again the unique $G^{\nu',S'}_{b'}$ that is in $\Xi$. By induction hypothesis we have that $\gamma(\Xi)(g') = \chi(g')$ for every input $g'$ of $g$.
We now distinguish two cases:
\begin{itemize}
	\item $b' = \root(T)$.
		Therefore
		by Lemma~\ref{lem:respects_semantics} we know that $\gamma(\Xi)$ respects the semantics of $g$, which means that  
		$\gamma(\Xi)(g) = \bigodot \gamma(\Xi)(W(g)) = \bigodot \chi(W(g)) = \chi(g)$ (where the second equality comes from the induction hypothesis
		and the third equality is just the definition of the evaluation $\chi$ of $C$), which proves the claim.
	\item $b' < \root(T)$.
          But then by Lemma~\ref{lem:suspicious_not_innocent_propagate} we must have $g \notin S'$
		 (because otherwise we should have $b' = \root(T)$ and then $S' = \emptyset$), that is $g$ is innocent for $G^{\nu',S'}_{b'}$.
		Therefore, again by Lemma~\ref{lem:respects_semantics}, it must be the case that
		$\gamma(\Xi)$ respects the semantics of $g$, and we can again show that $\gamma(\Xi)(g) = \chi(g)$, concluding the proof. \qedhere
\end{itemize}
\end{proof}

This indeed implies that $D \implies C$: let $\chi$ be a valuation of the variable gates and suppose $\chi(D) = 1$. Then by definition of the output of 
$D$, it means that the gate $G^{\nu,\emptyset}_{\root(T)}$ such that $\nu(g_\out)=1$ evaluates to $1$ under $\chi$.
But then, considering a trace $\Xi$ of $D$ starting at $G^{\nu,\emptyset}_{\root(T)}$ according to $\chi$, we have that $\chi(g_\out) = \gamma(\Xi)(g_\out) = \nu(g_\out) = 1$.
To show the converse ($C \implies D$), one can simply observe the following phenomenon:

\begin{lemma}
	\label{lem:follows2}
	Let $\chi$ be a valuation of the variable gates. Then for every
        bag $b \in T$, the gate~$G^{\chi_{|b},S}_b$ evaluates to $1$
        under $\chi$, where $S$ is the set of gates $g \in \UNF(\nu)$
        such that for all $g'$ input of $g$ that appears in some bag $b'
        \leq b$, then $\chi(g')$ is
        a weak value of $g$.
\end{lemma}
\begin{proof}
	Easily proved by bottom-up induction.
\end{proof}

Now suppose $\chi(C)=1$. By Lemma~\ref{lem:follows2} we have that
$G^{\chi_{|\root(T)},\emptyset}_{\root(T)}$ evaluates to $1$ under~$\chi$, and because $\chi(g_\out)=1$ we have that $\chi(D)=1$.
Hence, we have proved that $D$ is equivalent to $C$.

\subsection{$D$ is Deterministic}
We now prove
that $D$ is deterministic, i.e., that every $\OR$ gate in $D$ is deterministic.
Recall that the only $\OR$ gates in $D$ are the gates of the form $G^{\nu,S}_{b,\children}$.
We will in fact prove that traces are unique, which clearly implies that all the
$\OR$ gates are deterministic. 

We just need to prove the following lemma to reach our goal:

\begin{lemma}
\label{lem:innocent_or_suspect_but_not_innocent}
	Let $\chi$ be a valuation of the variable gates, $G^{\nu,S}_b$ a gate in $D$ that evaluates to~$1$ under $\chi$ and $\Xi$ be a 
	trace of $D$ starting at $G^{\nu,S}_b$ according to $\chi$. Let $g \in b$. Then the following is true:
	\begin{itemize}
		\item if $g$ is innocent ($g \notin S$) and $\nu(g)$ is a strong value of $g$,
			then there exists an input $g'$ of $g$ such that $\gamma(\Xi)(g)$ is a strong value for $g$.
		\item if $g \in S$, then for every input $g'$ of $g$ that is in the domain of $\gamma(\Xi)$, we have that $\gamma(\Xi)(g')$ is a weak value for $g$. 
	\end{itemize}
\end{lemma}
\begin{proof}
  We prove the two claims independently:
	\begin{itemize}
		\item Let $g \in b$ such that $g \notin S$ and $\nu(g)$ is a strong value for $g$.
			Then the claim directly follows from the second item of Lemma~\ref{lem:behaviour}.
                      \item We prove the second claim via a bottom-up induction on $T$.
			When $b$ is a leaf then it is trivially true because $g$ has no input $g'$ in $b$ because $|b| \leq 1$ because
			$T$ is nice.
			For the induction case, let $G^{\nu_l,S_l}_{b_l}$ and $G^{\nu_r,S_r}_{b_r}$ be the (unique) gates in $\Xi$ corresponding to the children
			$b_l, b_r$ of~$b$. By hypothesis we have $g \in S$.
			By definition of a gate being suspicious, we know that $\nu(g)$ is a strong value for $g$.
			To reach a contradiction, assume that there is an input $g'$ of $g$ in the domain of $\gamma(\Xi)$ such that
			$\gamma(\Xi)(g')$ is a strong value for $g$. Clearly this $g'$ is not in $b$, because $g$ is unjustified by $\nu$ (because $S \subseteq \UNF(\nu)$).
			Either $g'$ occurs in a bag $b_l' \leq b_l$, or it occurs in a bag $b_r' \leq b_r$. 
			The two cases are symmetric, so we assume that we are in the former.
                        As $g \in b$ and $g' \in b_l'$, by the properties of tree decompositions and because $g' \notin b$,
			we must have $g \in b_l$.
			Hence, by the contrapositive of the induction hypothesis on $b_l$ applied to $g$, we deduce that $g \notin  S_l$.
			But then by the second item of Definition~\ref{def:result}, $g$ should be innocent for $G^{\nu,S}_b$, that is $g \notin S$,	
			which is a contradiction.\qedhere
	\end{itemize}
\end{proof}

We are ready to prove that traces are unique.
Let us first introduce some useful notations:
Let $U$, $U'$ be sets, $\nu$, $\nu'$ be valuations. We write $(\nu,U) = (\nu',U')$ to mean $\nu = \nu'$ and $U = U'$, and we write 
$(\nu,U)(g) = (\nu',U')(g)$ to mean that $\nu(g)=\nu'(g)$ and that we have $g \in U$ iff $g \in U'$.
We show the following:

\begin{lemma}
	Let $\chi$ be a valuation of the variable gates such that $G^{\nu,S}_{b,\children}$ evaluates to $1$ under $\chi$. 
	Then there is a unique trace of $D$ starting at 
	$G^{\nu,S}_{b,\children}$ according to $\chi$.
\end{lemma}
\begin{proof}
	We will prove the claim by bottom-up induction on $T$. 
	The case when $b$ is a leaf is vacuous because there are no gates of the form $G^{\nu,S}_{b,\children}$ for a leaf.
	For the inductive case, let $b$ be an internal bag with children $b_l$ and $b_r$.
	By induction hypothesis for every $G^{\nu_l,S_l}_{b_l,
        \children}$ (resp., $G^{\nu_r,S_r}_{b_r, \children}$) 
	that evaluates to $1$ under $\chi$ there exists a unique trace
        $\Xi_l$ (resp., $\Xi_r$)
	of $D$ starting at $G^{\nu_l,S_l}_{b_l, \children}$ (resp., $G^{\nu_r,S_r}_{b_r, \children}$).
	Hence, if by contradiction there are more than two traces of $D$
        starting at $G^{\nu,S}_{b,\children}$, it can only be because
        $G^{\nu,S}_{b,\children}$ is not deterministic, i.e., because
	at least two different inputs of $G^{\nu,S}_{b,\children}$ 
	evaluate to $1$ under $\chi$,
	say $G^{\nu_l,S_l,\nu_r,S_r}_b$ and $G^{\nu'_l,S'_l,\nu'_r,S'_r}_b$ with $(\nu_l,S_l) \neq (\nu'_l,S'_l)$ or $(\nu_r,S_r) \neq (\nu'_r,S'_r)$. 
	W.l.o.g.\ we can suppose that it is $(\nu_l,S_l) \neq (\nu'_l,S'_l)$. Hence there exists $g_0 \in b_l$ such that 
	$(\nu_l,S_l)(g_0) \neq (\nu'_l,S'_l)(g_0)$.
	Let $\Xi_l$ be the trace of $D$ starting at $G^{\nu_l,S_l}_{b_l}$ and $\Xi'_l$ be the trace of $D$ starting at $G^{\nu'_l,S'_l}_{b_l}$.
	We observe the following simple fact about $\Xi_l$ and $\Xi'_l$:
	\begin{enumerate}
		\item[(*)] for any $g$, if $\gamma(\Xi_l)(g) \neq \gamma(\Xi'_l)(g)$ then $g \notin b$. 
			Indeed otherwise we should have $\nu(g) = \nu_l(g) = \gamma(\Xi_l)(g)$ and $\nu(g) = \nu'_l(g) = \gamma(\Xi'_l)(g)$, 
			which is impossible.
	\end{enumerate}
	Now we will define an operator $\theta$ that takes as input a gate $g$ such that $(\gamma(\Xi_l), S_l)(g) \neq (\gamma(\Xi'_l), S'_l)(g)$, and outputs another
	gate $\theta(g)$ which is an input of $g$ and such that again $(\gamma(\Xi_l), S_l)(\theta(g)) \neq (\gamma(\Xi'_l), S'_l)(\theta(g))$.
	This will lead to a contradiction because for any $n \in \NN$, starting with $g_0$ and applying $\theta$ $n$ times consecutively
	we would obtain $n$ mutually distinct gates (because $C$ is acyclic), but
	$C$ has a finite number of gates.

	Let us now prove that $\theta$ exists: let $g$ such that $(\gamma(\Xi_l), S_l)(g) \neq (\gamma(\Xi'_l), S'_l)(g)$.
	We distinguish two cases:
	\begin{itemize}
		\item We have $(\gamma(\Xi_l), S_l)(g) \neq (\gamma(\Xi'_l), S'_l)(g)$ because $\gamma(\Xi_l)(g) \neq \gamma(\Xi'_l)(g)$. Then by (*), 
			we know for sure that $g \notin b$.
			Therefore the topmost bag $b'$ in which $g$ occurs is $\leq b_l$.
			Let $G^{\nu',S'}_{b'}$ be the gate in $\Xi_l$ and
                        $G^{\nu'',S''}_{b'}$ the gate in $\Xi'_l$ (they exist and are unique by Lemma~\ref{lem:only_one_agree2}). 
			Then by Lemma~\ref{lem:suspicious_not_innocent_propagate} we must have $g \notin S'$ 
			and $g \notin S''$, because otherwise
			we should have $b'=b$, which is not true.
			Hence, by Lemma~\ref{lem:respects_semantics}
			 we know that both $\gamma(\Xi_l)$ and $\gamma(\Xi'_l)$ respect the semantics of $g$.
			But we have $\gamma(\Xi_l)(g) \neq \gamma(\Xi'_l)(g)$, so there must exist an input $g'$ of~$g$ such that $\gamma(\Xi_l)(g') \neq \gamma(\Xi'_l)(g')$!
			We can thus take $\theta(g)$ to be $g'$.
		\item We have $(\gamma(\Xi_l), S_l)(g) \neq (\gamma(\Xi'_l),
                  S'_l)(g)$ because (w.l.o.g.) $g \notin S_l$ and $g \in S'_l$.
			Observe that this implies that $g \in b_l$, and that $\nu'_l(g)$ is a strong value for $g$.
			We can assume that $\nu_l(g)=\nu'_l(g)$, as otherwise we would have $\gamma(\Xi_l)(g) \neq \gamma(\Xi'_l)(g)$, which is a case already
			covered by the last item. Hence $\nu_l(g)$ is also a strong value for $g$, but we have $g \notin S_l$, 
			so by the first item of Lemma~\ref{lem:innocent_or_suspect_but_not_innocent} we know that
			there exists an input $g'$ of $g$ that occurs in some bag $\leq b_l$ and such that $\gamma(\Xi_l)(g')$ is a strong value for $g$.
			We show that $\gamma(\Xi'_l)(g')$ must in contrast be a weak value for $g$, so that we can take $\theta(g)$ to be $g'$ and conclude the proof.
			Indeed suppose by way of contradiction that $\gamma(\Xi'_l)(g')$ is a strong value for~$g$. By the contrapositive of the second item
			of Lemma~\ref{lem:innocent_or_suspect_but_not_innocent}, we get that $g \notin S'_l$, which contradicts our assumption.
	\end{itemize}
	Hence we proved that $\theta$ exists, which shows a contradiction, which means that in fact we must have $G^{\nu_l,S_l,\nu_r,S_r}_b = G^{\nu'_l,S'_l,\nu'_r,S'_r}_b$, so that 
	$G^{\nu,S}_{b,\children}$ is deterministic, which proves that there is a unique trace of $D$ starting at $G^{\nu,S}_{b,\children}$ according to $\chi$, which
	was our goal.
\end{proof}
This concludes the proof that $D$ is deterministic, and thus that $D$ is a d-SDNNF equivalent to $C$.

\subsection{Analysis of the Running Time}
We last check that the construction can be performed in time $O(|T| \times f(k))$, where $f$ is in $O(2^{(4+\epsilon)k})$ for any $\epsilon > 0$:
\begin{itemize}
	\item From the initial tree decomposition $T$ of $C$, we paid
		$O(k|T|)$ to compute the nice tree decomposition $T_{\mathrm{nice}}$ of size
          $O(k|T|)$;
	\item We computed the mapping $\IT_{\mathrm{nice}}$ in linear time in $T_{\mathrm{nice}}$;
	\item We can clearly compute the v-tree in linear time from $T_{\mathrm{nice}}$;
	\item For each bag $b$ of $T_{\mathrm{nice}}$ we have $2^{2|b|}
          \leq 2^{2k+2}$ different pairs of a valuation $\nu$ of $b$ and of a subset $S$ of $b$, and checking 
if $\nu$ is a $(C,b)$-almost-evaluation and if $S$ is a subset of the unjustified gates of $\nu$ can be done in
polynomial time in $|b| \leq k+1$ (we access the inputs and the type of each
    gate in RAM, i.e., in  constant time, from $C$), hence we pay
    $O(|T_{\mathrm{nice}}| \times p(k) \times 2^{2k})$ to create the
    gates of the form $G^{\nu,S}_b$, for $p$ some polynomial;
	\item We pay an additional $O(|T_{\mathrm{nice}}| \times 2^{4k})$ to create the gates of the form $G_b^{\nu_l,S_l,\nu_r,S_r}$;
	\item We pay an additional $O(|T_{\mathrm{nice}}| \times p'(k) \times 2^{4k})$ to connect the gates of the form $G^{\nu,S}_{b, \children}$ to their inputs ($p'$ being again
		some polynomial).
\end{itemize}
Hence the total cost is indeed in $O(|T| \times f(k))$, where $f$ is in $O(2^{(4+\epsilon)k})$ for any $\epsilon > 0$.
\end{toappendix}

\section{Lower Bounds on OBDDs}
\label{sec:obddlower}
We now move to lower bounds on the size of structured representations
of Boolean functions, in terms of the width of a circuit for that
function.
Our end goal is to obtain a lower bound for (d-)SDNNFs,
that will form a counterpart to the upper bound
of Theorem~\ref{thm:upper_bound}.
We will do so in Section~\ref{sec:sddnnflower}. For now, in this section,
we consider a weaker
class of lineage representations than (d-)SDNNFs, namely, \emph{OBDDs}.

\begin{definition}
	\label{def:OBDD}
	An \emph{ordered binary decision diagram} (or \emph{OBDD}) on a set of
        variables $V = \{v_1, \ldots, v_n\}$
        is a rooted DAG $O$ whose leaves are labeled by $0$ or $1$, and whose
	internal nodes are labeled with a variable of~$V$ and have two outgoing edges labeled $0$ and $1$. 
	We require that there exists a total order $\mathbf{v} = v_{i_1}, \ldots,
	v_{i_n}$ on the variables such that, for every path from the root to a leaf,
        the sequence of variables which labels the internal nodes of the path
        is a subsequence of~$\mathbf{v}$ and does not contain duplicate
        variables.
        The OBDD $O$ \emph{captures} a Boolean function on~$V$ defined by mapping
        each valuation $\nu$ to the value of the leaf reached from the root by following
	the path given by~$\nu$.
	The \emph{size} $\card{O}$ of~$O$ is its number of nodes, and 
        the \emph{width} $w$ of~$O$ is the maximum number of nodes at every \emph{level}, where
        a level is defined for a prefix of~$\mathbf{v}$ as the set of nodes 
	reached by enumerating all possible valuations of this prefix.
        Note that we clearly have $\card{O} \leq \card{V} \times w$.
\end{definition}

Our upper bound in the previous section applied to arbitrary
Boolean circuits; however,
our lower bounds in this section and the next one will already apply to
much weaker formalisms for Boolean functions, namely, monotone DNFs and monotone CNFs (recall their definition from
Section~\ref{sec:preliminaries}). 
Some lower bounds are already known for the compilation of
monotone CNFs 
into OBDDs: 
Bova and Slivovsky have
constructed a family of CNFs of bounded degree whose OBDD width is exponential in
their number of variable occurrences~\cite[Theorem~19]{bova2017compiling},
following 
an earlier result of this type by Razgon~\cite[Corollary~1]{razgon2014obdds}. 
The result is as follows:

\begin{theorem}[{\cite[Theorem~19]{bova2017compiling}}]
\label{thm:obddlowerbova}
There is a class of monotone CNF formulas of bounded degree and arity such that every formula $\phi$ in this class has OBDD size at least 
$2^{\Omega(|\phi|)}$.
\end{theorem}

We adapt some of these techniques to show a more general result: our lower
bound applies to \emph{any} monotone DNF or monotone CNF, not to one
specific family.
Specifically, we show:

\begin{theorem}
  \label{thm:obddlower}
  Let $\phi$ be a monotone DNF (or monotone CNF), let $a\colonequals\arity(\phi)$ and
  $d\colonequals\degree(\phi)$. Then any OBDD for $\phi$ has width
  $\geq 2^{\left\lfloor\frac{\pw(\phi)}{a^3 \times d^2}\right\rfloor}$.
\end{theorem}

From our Theorem~\ref{thm:obddlower}, we can easily 
derive Theorem~\ref{thm:obddlowerbova} using the fact
(also used in the proof of~\cite[Theorem~19]{bova2017compiling})
that there exists a family of monotone CNFs of 
bounded degree and arity whose treewidth (hence pathwidth)
is linear in their size, namely,
the CNFs built from \emph{expander graphs} (see 
\cite[Theorem~5 and Proposition~1]{grohe2009treewidth}).
Note that expander graphs can also be used
to show lower bounds for \emph{(non-deterministic and non-structured) DNNFs} 
for a CNF formula~\cite{bova2015strongly}; our lower bound on SDNNFs of
Section~\ref{sec:sddnnflower} does not capture this result (because we need structuredness).

We observe that,
for a family of formulas with bounded arity and degree, the bound of
Theorem~\ref{thm:obddlower} is optimal, 
up to constant factors in the exponent. Indeed, following
earlier work \cite{ferrara2005treewidth,razgon2014obdds}, 
Bova and Slivovsky have shown that any CNF $\phi$ can be
compiled to OBDDs of width~$2^{\pw(\phi)+2}$ \cite[Theorem~4 and Lemma~9]{bova2017compiling}.
(Their upper bound result also applies to DNFs, and does not assume monotonicity nor a bound on
the arity or degree.)
In other words, for any monotone DNF or monotone CNF of bounded arity and
degree, pathwidth \emph{characterizes} the width of an OBDD for the formula, in
the following sense:

\begin{corollary}
  \label{cor:obdddnf}
  For any constant~$c$, for any monotone DNF (or monotone CNF) $\phi$ with arity
  and degree bounded by~$c$, the  smallest width of an OBDD for~$\phi$ is $2^{\Theta(\pw(\phi))}$.
\end{corollary}

This corollary talks about the pathwidth of~$\phi$ measured as that of its
hypergraph, but note that the same result would hold when measuring the
pathwidth of the incidence graph or dual hypergraph of~$\phi$. Indeed, 
all these
pathwidths are within a constant factor of one another when the degree and arity
are bounded by a constant.

We prove Theorem~\ref{thm:obddlower} in the rest of this section. We present the
proof in the case of monotone DNFs to reuse existing lower bound techniques
from~\cite{amarilli2016tractable,amarilli2016leveraging}, but explain at the end of this section
how the proof adapts to monotone CNFs.
We first present \emph{pathsplitwidth}, a new notion
which intuitively measures the performance of a variable ordering for an OBDD
on the monotone DNF~$\phi$, and connect it to the
pathwidth of~$\phi$. Second, we
recall the definition of \emph{dncpi-sets} introduced
in~\cite{amarilli2016tractable,amarilli2016leveraging} to show lower bounds from
the structure of Boolean
functions. Last, we conclude the proof by
connecting pathsplitwidth to the size of dncpi-sets.

\myparagraph{Pathsplitwidth} The first step of the proof is to rephrase the
bound on pathwidth, arity, and degree, in terms of a bound on the performance of variable
orderings. Intuitively, a good variable ordering is one which does not \emph{split}
too many clauses. Formally:

\begin{definition}
	\label{def:pathsplitwidth}
	Let $H=(V,E)$ be a hypergraph, and
        $\mathbf{v} = v_1, \ldots, v_{|V|}$ be an ordering on the variables of~$V$.
        For $1 \leq i \leq \card{V}$,
        we define $\spl_i(\mathbf{v}, H)$ as the set of hyperedges $e$ of~$H$ that
        contain both a variable at or before $v_i$, and a variable strictly after~$v_i$, formally:
        $\spl_i(\mathbf{v}, H) \colonequals \{e \in E \mid \exists l \in \{1,
        \ldots, i\} \text{~and~}
        \exists r \in \{i+1, \ldots, \card{V}\} \text{~such that~} \{v_l, v_r\} \subseteq e\}$.
Note that $\spl_{|V|}(\mathbf{v},H)$ is always empty.
        The \emph{pathsplitwidth} of $\mathbf{v}$
        relative to $H$ is the maximum size of the split, formally,
        $\psw(\mathbf{v}, H) \colonequals \max_{1 \leq i \leq |V|} |
        \spl_i(\mathbf{v},H) |$.
        The \emph{pathsplitwidth} $\psw(H)$ of $H$ is then the
        minimum
        of $\psw(\mathbf{v}, H)$
        over all variable orderings $\mathbf{v}$ of~$V$.
\end{definition}

In other words, $\psw(H)$ is the smallest integer $n\in \NN$ such that,
for any variable ordering~$\mathbf{v}$ of the nodes of $H$, there is a moment at which $n$
hyperedges of $H$ are \emph{split}, i.e., for $n$ hyperedges~$e$, we have begun 
enumerating the nodes of $e$ but we have not enumerated all of them yet. 
We note that the pathsplitwidth of $H$ is exactly the \emph{linear
branch-width}~\cite{exploring2017nordstrand} of the dual hypergraph of $H$,
but we introduced pathsplitwidth because it
fits our proofs better.

For a
monotone DNF $\phi$ with hypergraph $H$, the quantity $\psw(H)$ is intuitively
connected to the quantity of information that an OBDD will have to remember when
evaluating~$\phi$ following any variable ordering, which we will formalize via
dncpi-sets. This being said, the definition of pathsplitwidth is also reminiscent of
that of pathwidth, and we can indeed connect the two (up to a factor of the
arity):

\begin{lemmarep}
	\label{lem:cw_pw}
        For any hypergraph $H=(V,E)$, we have
        $\pw(H) \leq \arity(H) \times \psw(H)$.
\end{lemmarep}

\begin{proofsketch}
  From a variable ordering $\mathbf{v}$, we construct a path decomposition of~$H$
  by creating $\card{V}$ bags in sequence, each of which containing $v_i$ plus
  $\bigcup \spl_i(\mathbf{v}, H)$.
  The width is
  $\leq \arity(H) \times \psw(H)$, and we check the two conditions of path
  decompositions.
  First, each
  hyperedge of~$H$ is contained in a bag where it is split. Second, each
  vertex $v_i$
  occurs in the corresponding bag $b_i$ and at all positions where the edges
  containing~$v$ are split, which forms a 
  segment of~$\mathbf{v}$:
  thus, the connectedness condition of tree decompositions is respected.
\end{proofsketch}

\begin{proof}
  Let $H = (V, E)$ be a hypergraph, and
  let $\mathbf{v}$ be an enumeration of the nodes of $H$ witnessing that $H$ has
  pathsplitwidth $\psw(H)$.
  We will construct a path decomposition of~$H$ of width $\leq \arity(H) \times
  \psw(H)$.
  Consider the path $P = b_1, \cdots, b_{|V|}$ and the labeling function
  $\lambda$
  where $\lambda(b_i) \colonequals \{v_i\} \cup \bigcup
  \spl_i(\mathbf{v},H) $ for $1 \leq i \leq |V|$. Let us show that $(P,
  \lambda)$ is a path 
  decomposition of $H$: once this is established, it is clear that its width
  will be $\leq \arity(H) \times \psw(H)$.

  First, we verify the occurrence condition. Let $e \in E$.
  If $e$ is a singleton $\{v_i\}$ then $e$ is included in $b_i$. 
Now, if $|e| \geq 2$, then let $v_i$ be the first element of~$e$
  enumerated by~$\mathbf{v}$. We have $e\in\spl_i(\mathbf{v},H)$,
  and therefore $e$ is included in~$b_i$.

Second, we verify the connectedness condition. Let $v$ be a vertex of $H$,
  then by definition $v \in b_i$ iff $v=v_i$ or there exists $e \in \spl_i(\mathbf{v},H)$
with $v \in e$. We must show that the set~$T_v$ of the bags that contain
  $v$ forms a connected subpath in
  $P$. To show this, first observe that for every $e \in E$, letting $\spl(e) =
  \{v_i\mid 1\leq i<|V| \land e \in \spl_{i}(\mathbf{v},H)\}$, then $\spl(e)$ is clearly a
  connected segment of $\mathbf{v}$. Second, note that for every $e$ with $v \in
  e$, 
  then either $v \in \spl(e)$ or $v$ and the connected subpath $\spl(e)$ are adjacent (in the case where $v$
  is the last vertex of~$e$ in the enumeration). Now, by definition $T_v$
  is the union of the $b_{v'}$ for $v' \in \spl(e)$ with $v \in e$ and of~$b_i$, so it is a 
  union of connected subpaths which all contain $b_i$ or are adjacent to it:
  this establishes that $T_v$ is a connected subpath, which shows in turn that $(T,
  \lambda)$ is a path decomposition, concluding the proof.
\end{proof}

\begin{toappendix}
  For completeness with the preceding result, we note that the following also
  holds, although we do not use it in the proof of Theorem~\ref{thm:obddlower}:

\begin{lemma}
        For any hypergraph $H=(V,E)$, we have $\psw(H) \leq \degree(H) \times (\pw(H) + 1)$.
\end{lemma}

  \begin{proof}
    Let $P = b_1 -\; \cdots\;- b_m$ be a path decomposition of $H$ of width $\pw(H)$.
For $1 \leq i \leq m$ we define $\first(b_i)$ to be the set of all $v\in
    b_i$ such that $b_i$ is the first bag containing $v$ (this set can be
    empty). Let $\mathbf{v}_i$ be any
ordering on $\first(b_i)$. Consider the ordering $\mathbf{v} \defeq
    \mathbf{v}_1 \ldots \mathbf{v}_m \mathbf{v}'$, where $\mathbf{v}'$ is
    any ordering of the remaining vertices of~$H$ (i.e., those that do
    not appear in~$P$ because they are not present in any hyperedge). Let
    $n=|\mathbf{v}|=|H|$.
We claim that for any $1\leq i\leq m$, we have $|\spl_i(\mathbf{v},H)| \leq \degree(H) \times (\pw(H)+1)$, which clearly implies that $\psw(H) \leq \degree(H) \times (\pw(H)+1)$.
    This is clear for $v_i$ in $\mathbf{v}'$, since then
    $|\spl_i(\mathbf{v},H)|=0$. 
Now suppose
    $v_i \in \mathbf{v}_j$ for some $1 \leq i \leq m$. Let $e \in \spl_i(\mathbf{v},H)$. 
We will show that $\exists v' \in b_j$ such that $v' \in e$, which will imply that $|\spl_i(\mathbf{v},H)| \leq \degree(H) \times (\pw(H)+1)$.
Assume by way of contradiction that there is no such $v'$.
We know that $e \in \spl_i(\mathbf{v},H)$, hence there exist $v^-,v^+$ with
    $\{v^-,v^+\} \subseteq e$ and $v^- \in \mathbf{v}_{j^-}$ for some $j^- < j$
    and $v^+ \in {\mathbf v}_{j^+}$ for some $j^+ > j$.
But, as $P$ is a path decomposition of $H$ and $\{v^-,v^+\} \subseteq e$, $v^-$ and $v^+$ must appear together in a bag! 
Now, as $b_{j^+}$ is the first bag in which $v^+$ appears, it must be the
    case that $v^- \in b_{j^+}$, and therefore $v^-\in b_j$ (otherwise
    the connectedness property would be violated), which leads to a
    contradiction and concludes the proof.
  \end{proof}
\end{toappendix}

Thanks to Lemma~\ref{lem:cw_pw}, it suffices to show that an OBDD for~$\phi$ has width
$\geq 2^{\left\lfloor\frac{\psw(\phi)}{(a\times d)^2}\right\rfloor}$, which we will do in
the rest of this section.

\myparagraph{dncpi-sets} To show this lower bound, we use the technical
tool of \emph{dncpi-sets}
\cite{amarilli2016tractable,amarilli2016leveraging}. We recall the
definitions here, adapting the notation
slightly.
Remember that our monotone DNFs are assumed to be minimized.
Note that dncpi-sets are reminiscent of 
\emph{subfunction width} in~\cite{bova2017compiling} (see Theorem~17
in~\cite{bova2017compiling}),
but
the latter notion is only defined for graph CNFs.

\begin{definition}[{\cite[Definition~6.4.6]{amarilli2016leveraging}}]
	\label{def:dncpi}
  Given a monotone DNF $\phi$ on variables $V$,
  a \emph{disjoint non-covering prime implicant set} (dncpi-set) of $\phi$ is a set $S$
  of clauses of~$\phi$ which:
  \begin{itemize}
    \item are pairwise disjoint: for any $D_1 \neq D_2$ in $S$, we have
      $D_1 \cap D_2 = \emptyset$.
    \item are \emph{non-covering} in the following sense: for any clause
      $D$ of $\phi$, if $D \subseteq \bigcup S$, then $D \in
      S$.
  \end{itemize}
  The \emph{size} of~$S$ is the number of clauses that it contains.

  Given a variable ordering $\mathbf{v}$ of~$V$,
  we say that $\mathbf{v}$ \emph{shatters} a dncpi-set $S$
  if
  there exists $1\leq i \leq \card{V}$ such that $S \subseteq \spl_i(\mathbf{v},H)$,
  where $H$ is the hypergraph of $\phi$.
\end{definition}

Observe the analogy between shattering and splitting, which we will substantiate
below. We recall the main result on dncpi-sets:

\begin{lemma}[{\cite[Lemma~6.4.7]{amarilli2016leveraging}}]
  \label{lem:dncpi-shattered}
  Let $\phi$ be a monotone DNF on variables~$V$ and $n \in \NN$. Assume that, for
  every variable ordering~$\mathbf{v}$ of~$V$, there is some 
  dncpi-set $S$ of $\phi$ with $\card{S} \geq n$, such that $\mathbf{v}$ shatters $S$.
  Then any OBDD for $\phi$ has width $\geq 2^n$.
\end{lemma}

\begin{proofsketch}
  Considering the point at which the dncpi-set is shattered, the OBDD must
  remember exactly the status of each clause of the set: any valuation that
  satisfies a subset of these clauses gives rise to a different continuation
  function. This is where we use the fact that the DNF is monotone: it ensures
  that we can freely choose a valuation of the variables that do not occur in
  the dncpi-set without making the formula true.
\end{proofsketch}

\myparagraph{Concluding the proof}
We conclude the proof of Theorem~\ref{thm:obddlower} by showing that
any variable ordering of the variables of a monotone DNF $\phi$ shatters a dncpi-set of the right
size. The formal statement is as follows, and it is the last result to prove:

\begin{lemma}
	\label{lem:exists-dncpi}
        Let $\phi$ be a monotone DNF, $H$ its hypergraph, and 
        $\mathbf{v}$ an enumeration of its variables.
        Then there is a dncpi-set $S$ of $\phi$ shattered by $\mathbf{v}$ such
        that $|S| \geq \left\lfloor\frac{\psw(H)}{(\arity(H) \times
        \degree(H))^2}\right\rfloor$.
\end{lemma}

We prove this result in the rest of the section.
Our goal is to construct a dncpi-set, which intuitively consists of clauses that
are disjoint and which do not cover another clause. We can do so by
picking clauses sufficiently ``far apart''. Let
the \emph{exclusion graph} of~$H = (V, E)$ be the graph on~$E$ where two edges
$e \neq e'$ are
adjacent if 
there is an edge $e''$ of~$E$ with which they both share a node: this is in 
particular the case when $e$ and $e'$ intersect as we can take $e'' \colonequals
e$.
Formally, the exclusion graph is $G_H = (E, \{\{e, e'\} \in E^2 \mid e \neq e'
\land \exists e''
\in E, (e \cap e'') \neq \emptyset \wedge (e' \cap e'') \neq \emptyset\})$. In
other words, two hyperedges are adjacent in~$G_H$ iff they are different and are at distance at most~4 in the
incidence graph of~$H$.

Remember that an
\emph{independent set} in the graph~$G_H$ is a subset $S$ of~$E$ such that
no two elements of~$S$ are adjacent in $G_H$. The definition of~$G_H$ then ensures:

\begin{lemmarep}
  \label{lem:exclusion}
  For any monotone DNF $\phi$, letting $H$ be its hypergraph, any independent
  set of the exclusion graph~$G_H$ is a dncpi-set of~$\phi$.
\end{lemmarep}

\begin{proof}
  The vertices of~$G_H$ are clauses of~$\phi$ by construction.
  Now, the elements of an independent set $S$
  are pairwise disjoint clauses, because whenever two clauses $e$ and $e'$
  intersect, then taking $e'' \colonequals e$, we have that $e''$ intersects
  both $e$ and~$e'$, so there is an edge between $e$ and~$e'$ in the exclusion
  graph, so $e$ and $e'$ cannot both occur in an independent set.
  Now, to show why $S$ is non-covering,
  assume by contradiction that there exists a clause $e''$ of~$\phi$ which is
  not in~$S$ and such that $e'' \subseteq \bigsqcup S$. Remember that $\phi$ has
  been minimized, so $e''$ cannot be a strict subset of a single clause of~$S$,
  and it cannot be a clause of~$S$ by hypothesis. Hence, there must be two
  clauses $e \neq e'$ in~$S$ such that $e''$ intersects both~$e$ and~$e'$. Thus,
  $e''$ witnesses that there is an edge between $e$ and~$e'$ in the exclusion
  graph, so they cannot be both part of~$S$, a contradiction. This concludes the
  proof.
\end{proof}

In other words, our goal is to compute a large independent set of the exclusion
graph. To do this, we will use the following straightforward lemma about independent
sets:

\begin{lemmarep}
  \label{lem:indepset}
  Let $G = (V, E)$ be a graph and let $V' \subseteq V$.
  Then $G$ has an independent set $S \subseteq V'$ of size at least
  $\left\lfloor \frac{\card{V'}}{\degree(G) + 1} \right\rfloor$.
\end{lemmarep}

\begin{proof}
  We construct the independent set $S$ with the following trivial algorithm: start
  with $S \colonequals \emptyset$ and, while $V'$ is non-empty, pick an arbitrary
  vertex $v$ in~$V'$, add it to~$S$, and remove $v$ and all its neighbors
  from~$G$ and from~$V'$. It is clear that this algorithm terminates and adds the prescribed
  number of vertices to~$S$, so all that remains is to show that $S$ is an
  independent set at the end of the algorithm. This is initially true for $S =
  \emptyset$; let us show that it is preserved throughout the algorithm. Assume
  by way of contradiction that, at a stage of the algorithm, we add a vertex $v$
  to~$S$ and that it stops being an independent set. This means that $S$
  contains a neighbor $v'$ of~$v$ which must have been added earlier; but when
  we added $v'$ to~$S$ we have removed all its neighbors from~$G$, so we have
  removed $v$ and we cannot add it later, a contradiction. Hence, the algorithm
  is correct and the claim is shown.
\end{proof}

Moreover, we can bound the degree of~$G_H$ using the degree and arity of~$H$:
\begin{lemmarep}
  \label{lem:exclusiondegree}
  Let $H$ be a hypergraph. Then $\degree(G_H) \leq (\arity(H) \times \degree(H))^2
  - 1$.
\end{lemmarep}

\begin{proofsketch}
  The bound on the arity and degree of~$H$ implies a bound on the number of
  edges that can be at distance $\leq 4$ of another edge in the incidence graph
  of~$H$, hence bounding the degree of the exclusion graph.
\end{proofsketch}

\begin{proof}
  Any edge~$e$ of~$H$ contains $\leq \arity(H)$ vertices, each of which occurs
  in $\leq \degree(H)-1$ edges that are different from~$e$, so any edge $e$
  of~$H$ intersects at most $n \colonequals \arity(H) \times (\degree(H)-1)$
  edges different from~$e$. Hence, the degree of~$G_H$ is at most $n + n^2$
  (counting the edges that intersect $e$ or those at distance~$2$ from~$e$).
  Now, we have $n + n^2 = n(n+1)$, and as $\degree(H) \geq 1$ and $\arity(H)
  \geq 1$ (because we assume that hypergraphs contain at least one non-empty
  edge), the degree of~$G_H$ is $< \arity(H) \times \degree(H) \times (1 +
  \arity(H) \times (\degree(H) - 1))$, i.e., it is indeed $< (\arity(H) \times
  \degree(H))^2$, which concludes.
\end{proof}

We are now ready to conclude the proof of Lemma~\ref{lem:exists-dncpi}:
\vspace{.5ex}

\begin{proof}[Proof of Lemma~\ref{lem:exists-dncpi}]
        Let $\phi$ be a monotone DNF,  $H=(V,E)$ its hypergraph, and $\mathbf{v}$~an enumeration of its variables.
	By definition of pathsplitwidth, there is $v_i \in V$ such that,
        for $E' \colonequals \spl_{i}(\mathbf{v},H)$, we have
        $|E'| \geq \psw(H)$.
	Now, by Lemma~\ref{lem:indepset}, $G_H$ has an independent set $S
        \subseteq E'$ of size 
	at least $\left\lfloor \frac{\card{E'}}{\degree(G_H) +
        1} \right\rfloor$ which is $\geq 
	\left\lfloor \frac{\psw(H)}{(\arity(H) \times \degree(H))^2} \right\rfloor$ by
        Lemma~\ref{lem:exclusiondegree}.
        Hence, $S$ is a dncpi-set by Lemma~\ref{lem:exclusion}, has the
        desired size, and is shattered since $S
        \subseteq E'$.
\end{proof}

Combining this result with Lemma~\ref{lem:cw_pw} and
Lemma~\ref{lem:dncpi-shattered} concludes the proof of
Theorem~\ref{thm:obddlower}.

\myparagraph{From DNFs to CNFs}
	We now argue that Theorem~\ref{thm:obddlower} also holds for monotone CNFs.
	Let $\phi$ be a monotone CNF, $a\colonequals\arity(\phi)$ and
  $d\colonequals\degree(\phi)$, and suppose for a contradiction that
  there is an OBDD $O$ for $\phi$ of width
  $< 2^{\left\lfloor\frac{\pw(\phi)}{a^3 \times d^2}\right\rfloor}$.
	Consider the monotone DNF $\phi'$ built from $\phi$ by replacing each $\land$ by a $\lor$ and each $\lor$ by a $\land$.
	Now, let $O'$ be the OBDD built from $O$ by replacing the label $b \in \{0,1\}$ of each edge by $1-b$, and replacing the label $b$ of each leaf by
	$1-b$. It is clear, by De Morgan's laws, that $O'$ is an OBDD for $\phi'$
	 of size $< 2^{\left\lfloor\frac{\pw(\phi)}{a^3 \times
         d^2}\right\rfloor}$, which contradicts Theorem~\ref{thm:obddlower} applied to monotone DNFs.

\section{Lower Bounds on d-SDNNFs}
\label{sec:sddnnflower}
In the previous section, we have shown that \emph{pathwidth} measures how concisely an
OBDD can represent a monotone DNF or CNF formula with bounded degree and arity.
In this section, we move from OBDDs to (d-)SDNNFs, and show that
\emph{treewidth} plays a similar role to pathwidth in this setting. Formally,
we show the following analogue of Theorem~\ref{thm:obddlower}:

\begin{theorem}
  \label{thm:dSDNNFlower}
  Let $\phi$ be a monotone DNF (resp., monotone CNF), let $a\colonequals\arity(\phi)$ and
  $d\colonequals\degree(\phi)$. Then any d-SDNNF (resp., SDNNF) for $\phi$ has size
  $\geq 2^{\left\lfloor\frac{\tw(\phi)}{3 \times a^3 \times
  d^2}\right\rfloor}-1$.
\end{theorem}

Combined with Theorem~\ref{thm:upper_bound} (or with existing results
specific to CNF formulas such as~\cite[Corollary 1]{bova2015compiling}), this yields an analogue of
Corollary~\ref{cor:obdddnf}. However, its statement is less neat:
unlike OBDDs, (d-)SDNNFs have no obvious notion of width, so the lower bound
above refers to size rather than width, and it does not exactly match our upper
bound. We obtain:
\begin{corollary}
  \label{cor:dsdnnfdnf}
  For any constant~$c$, for any monotone DNF (resp., monotone CNF) $\phi$
  with arity and degree bounded by~$c$, there
  is a d-SDNNF for~$\phi$ having size $\card{\phi} \times 2^{O(\tw(\phi))}$, and
  any d-SDNNF (resp., SDNNF) for~$\phi$ has size~$2^{\Omega(\tw(\phi))}$.
\end{corollary}

Our proof of Theorem~\ref{thm:dSDNNFlower} 
will follow the same overall structure as in the previous section.
We present the proof for monotone DNFs and d-DNNFs: see Appendix~\ref{apx:dSDNNFCNFs}
for the extension to monotone CNFs and SDNNFs.
Recall that d-SDNNFs are structured by~\emph{v-trees},
which generalize
variable orders. We
first introduce \emph{treesplitwidth}, a width notion that 
measures the performance of a v-tree 
by counting how many clauses it splits; and we connect treesplitwidth to treewidth.
We use again
dncpi-sets, and argue that a d-SDNNF structured by a v-tree must
shatter a dncpi-set whose size follows the treesplitwidth of the v-tree. We 
then show that shattering a dncpi-set forces d-SDNNFs to be large: instead of
the easy OBDD result of the previous section
(Lemma~\ref{lem:dncpi-shattered}), we will need a much deeper result of
Pipatsrisawat and Darwiche~\cite[Theorem~3]{pipatsrisawat2010lower}, rephrased
in the setting of communication complexity
by Bova, Capelli, Mengel, and Slivovsky~\cite{bova2016knowledge}.

Note that \cite{bova2016knowledge}, by a similar approach, shows an exponential
lower bound on the size of d-SDNNF which is reminiscent of ours. However, their bound again
applies to one well-chosen family of
Boolean functions; our contribution is to show a general lower bound. In
essence, our result is shown by observing that
the family of functions used in their lower bound occurs ``within'' any
bounded-degree, bounded-arity monotone DNF.
Also note that a result similar to the lower bound of Corollary~\ref{cor:dsdnnfdnf}
is proven by
Capelli~\cite[Corollary~6.35]{capelli2016structural}
as an auxiliary statement
to separate structured DNNFs and FBDDs.
The result
uses \emph{MIM-width},
but Theorem~4.2.5 of \cite{vatshelle2012new}, as degree and arity are
bounded, implies that we could rephrase it to
treewidth;
further, the result assumes arity-2 formulas, but it could be extended to arbitrary
arity as 
in~\cite[Theorem~12]{capelli2017understanding}.
More importantly, the result applies only to monotone CNFs and not to DNFs
.

\begin{toappendix}
  \subsection{Proof of Auxiliary Lemmas}
\end{toappendix}

\myparagraph{Treesplitwidth} Informally, treesplitwidth is to v-trees what
pathsplitwidth is to variable orders: it bounds the ``best performance'' of
any v-tree.
\begin{definition}
	\label{def:treesplitwidth}
	Let $H=(V,E)$ be a hypergraph, and $T$ be a v-tree over $V$. For any
        node~$n$ of $T$,
	we define $\spl_n(T,H)$ as the set of hyperedges $e$ of~$H$
        that contain both a variable in~$T_n$ and one outside~$T_n$
        (recall that $T_n$ denotes the subtree of~$T$ rooted
        at~$n$).
        Formally:
        \mbox{$\spl_n(T,H) \colonequals \{e \in E \mid \exists v_{\mathrm i} \in \LEAVES(T_n)
        \text{~and~} \exists
        v_{\mathrm o} \in \LEAVES(T \setminus T_n) \text{~such~that~}
        \{v_{\mathrm i}, v_{\mathrm o}\} \subseteq e\}$}.

	The \emph{treesplitwidth} of $T$ relative to $H$ is $\tsw(T, H)
        \colonequals \max_{n \in T} |\spl_n(T,H)|$.
        The \emph{treesplitwidth} $\tsw(H)$ of~$H$ is then the
        minimum
         of $\tsw(T, H)$
        over all v-trees $T$ of $V$.
\end{definition}

Again, the treesplitwidth of $H$ is exactly the
\emph{branch-width}~\cite{robertson1991obstructions} of the dual hypergraph of
$H$, but treesplitwidth is more convenient for our proofs.
As with pathsplitwidth and pathwidth (Lemma~\ref{lem:cw_pw}), we can
bound the treewidth of a hypergraph by its treesplitwidth:

\begin{lemmarep}
	\label{lem:tsw_tw}
	For any hypergraph $H=(V,E)$, we have
        $\tw(H) \leq 3 \times \arity(H) \times \tsw(H)$.
\end{lemmarep}

\begin{proofsketch}
  We construct a tree decomposition from a v-tree $T$: it has same skeleton
  as~$T$, its leaf bags contain the corresponding variable in the v-tree, and
  its internal bags contain the split at this v-tree node unioned with
  the split at the child
  nodes. This is indeed a tree decomposition because each non-singleton edge is
  split, and the nodes of the v-tree where a vertex of~$H$ occurs always form a connected
  subtree.
\end{proofsketch}

\begin{proof}
  Let $H = (V, E)$ be a hypergraph, and $T$ a v-tree over~$V$ witnessing
  that $H$ has treesplitwidth $\tsw(H)$.
	We will construct a tree decomposition $T'$ of $H$ of width $\leq 3
        \times 
        \arity(H) \times \tsw(H)$.
        The skeleton of $T'$ is the same as that of $T$. Now, for each node $n
        \in T$, we call $b_n$ the corresponding bag of $T'$, and we define the
        labeling $\lambda(b_n)$ of~$b_n$.

        If $n$ is an internal node of $T$ with children $n_l,n_r$ (recall
        that v-trees are assumed to be binary), then 
        we define
        $\lambda(b_n) \colonequals \bigcup \spl_n(T,H) \cup \bigcup \spl_{n_l}(T,H) \cup \bigcup \spl_{n_r}(T,H)$, 
	and if $n$ is a variable $v \in V$ (i.e., $n$ is a leaf of $T$) then
        $\lambda(b_n) \colonequals \{v\}$.
        It is clear that the width of $P$ is $\leq \max(3 \times \arity(H) \times
        \tsw(H),1) -1 \leq 3 \times \arity(H) \times \tsw(H)$.
	
        The occurrence condition is verified: let $e$ be an edge of
        $H$. If $e$ is a singleton edge~$\{v\}$
	then it is included in $b_v$. If $|e| \geq 2$ then there must exists a node $n \in T$ such that $e \in \spl_n(T,H)$.
	If $n$ is an internal node of $T$ then $e \subseteq \bigcup \spl_n(T,H) \subseteq b_n$, and if
	$n$ is a leaf node of $T$ then it must have a parent $p$ (since
        $e$ is split), and $e \subseteq \bigcup \spl_n(T,H) \subseteq b_p$.

        Connectedness is proved in the same way as in the proof of
        Lemma~\ref{lem:cw_pw}: for a given vertex $v\in V$, the nodes
        of~$T$ where each edge~$e$
        containing~$v$ is split is a connected subtree of~$T$ without its root
        node: more
        precisely, they are all the ancestors of a leaf in~$e$ 
        strictly lower than their the least
        common ancestor. Adding the missing root to each
        such subtree and unioning them all results in the subtree of all
ancestors of a vertex adjacent to~$v$ ($v$~itself included) up to their
  least common ancestor~$a$.
        Consequently, the set of nodes of~$T'$ containing~$v$ is
        a connected subtree of~$T'$, rooted in~$b_a$.
\end{proof}

Moreover, using the same techniques that we used in the last section, we can
show the analogue of Lemma~\ref{lem:exists-dncpi}.
Specifically, given a monotone DNF $\phi$ on variables $V$, a v-tree $T$ over $V$, and a dncpi-set $S$ of~$\phi$, we
say that $T$ \emph{shatters} $S$ if there is a node $n$ in $T$ such that $S
\subseteq \spl_n(T, \phi)$. We now show that any v-tree over $V$ must shatter a
large dncpi-set (depending on the treewidth, degree, and arity):

\begin{lemmarep}
	\label{lem:exists-dncpi2}
        Let $\phi$ be a monotone DNF, $H$ its hypergraph, and 
        $T$ be a v-tree over its variables.
        Then there is a dncpi-set $S$ of $\phi$ shattered by $T$ such
        that $|S| \geq \left\lfloor\frac{\tsw(H)}{(\arity(H) \times
        \degree(H))^2}\right\rfloor$.
\end{lemmarep}

\begin{proofsketch}
  The proof is just like that of Lemma~\ref{lem:exists-dncpi}, except with the
  new definition of split on v-trees. In particular, we use
  Lemmas~\ref{lem:exclusion}, \ref{lem:indepset}, and~\ref{lem:exclusiondegree}.
\end{proofsketch}
\begin{proof}
        Let $\phi$ be a monotone DNF,  $H=(V,E)$ its hypergraph, and
        $T$ a v-tree over~$V$.
        By definition of treesplitwidth, there exists $n \in V$ such that,
        letting $E' \colonequals \spl_{n}(T,H)$, we have
        $|E'| \geq \tsw(H)$.
	Now, by Lemma~\ref{lem:indepset}, $G_H$ has an independent set $S
        \subseteq E'$ of size 
	at least $\left\lfloor \frac{\card{E'}}{\degree(G_H) +
        1} \right\rfloor$ which is $\geq 
	\left\lfloor \frac{\tsw(H)}{(\arity(H) \times \degree(H))^2} \right\rfloor$ by
        Lemma~\ref{lem:exclusiondegree}.
        Hence, $S$ is a dncpi-set by Lemma~\ref{lem:exclusion}, has the
        desired size, and is shattered since $S
        \subseteq E'$.
\end{proof}
 
Hence, to prove Theorem~\ref{thm:dSDNNFlower}, the only missing
ingredient is a lower bound on the size of d-SDNNFs that shatter large
dncpi-sets. Specifically, we need an analogue of Lemma~\ref{lem:dncpi-shattered}:

\begin{lemmarep}
  \label{lem:dncpi-shattered2}
  Let $\phi$ be a monotone DNF on variables~$V$ and $n \in \NN$. Assume that, for
  every v-tree~$T$ over~$V$, there is some 
  dncpi-set $S$ of $\phi$ with $\card{S} \geq n$, such that $T$ shatters $S$.
  Then any d-SDNNF for $\phi$ has size $\geq 2^n-1$.
\end{lemmarep}

We will prove Lemma~\ref{lem:dncpi-shattered2} in the rest of this section using a recent lower bound by
Bova, Capelli, Mengel, and Slivovsky \cite{bova2016knowledge}.
They bound the size of any d-SDNNF for the \emph{set
intersection} function, defined as
$\SINT_n \defeq (x_1 \land y_1) \lor  \ldots \lor (x_n \land y_n)$.
This bound is useful for us:
a dncpi-set intuitively
isolates some variables on which $\phi$ computes exactly $\SINT_n$:

\begin{lemmarep}
	\label{lem:dncpi_fn}
	Let $\phi$ be a DNF with variables $V$, and
        let $S=\{D_1,\ldots,D_n\}$ be a dncpi-set of~$\phi$ 
        where every clause has size $\geq 2$.
        Pick two variables $x_i \neq y_i$ in~$D_i$ for each $1 \leq i \leq
        n$, and
        let $V' \colonequals \{x_1, y_1, \ldots, x_n, y_n\}$.
        Then there is a partial valuation $\nu$ of~$V$ with domain~$V \setminus V'$ such that
        $\nu(\phi) = \SINT_n$.
\end{lemmarep}
\begin{proofsketch}
	The valuation $\nu$ sets to~$1$ the variables $V''$ which are in~$\bigcup S$ but not in~$V'$, and
  sets to~$0$ all remaining variables. This amounts to discarding the clauses
  not in the dncpi-set, and discarding the variables of~$V''$ in the dncpi-set:
  what remains of the DNF is then precisely $\SINT_n$. Note that this result relies on
  monotonicity, and on the fact that $\phi$ is a DNF. (However, as we show in
  Appendix~\ref{apx:dSDNNFCNFs}, a dual result holds for monotone CNF.)
\end{proofsketch}
\begin{proof}
	Define the following partial valuation $\nu : V \setminus V' \to
        \{0,1\}$ that maps all the variables of $\bigcup_{i=1,\ldots,n} (D_i \setminus \{x_i,y_i\})$ to $1$ and
	all the other variables of $V \setminus \bigcup S$ to $0$.
	Let us show that for a clause $D \in \phi \setminus S$ we have $\nu(D)=0$. Otherwise, as all the variables that $\nu$ maps to $1$ are in $\bigcup_{i=1,\ldots,n} X_i \cup Y_i$,
	we should have $D \subseteq \bigcup S$, but because $S$ is a dncpi-set we should have $D \in S$ which is a contradiction.
	Now, $\nu$ maps all the variables of $D_i \setminus \{x_i,y_i\}$
        to~$1$, hence $\nu(\phi)$ indeed captures $\SINT_n$.
\end{proof}

\begin{toappendix}
\begin{proof}[Proof of Lemma~\ref{lem:dncpi-shattered2}]
	Let $C$ be a d-SDNNF structured by a v-tree $T$ that captures $\phi$. 
	Consider the dncpi-set $S=\{D_1,\ldots,D_m\}$ of size $\geq n$ of $\phi$ 
        that is shattered by $T$ (note that this implies in particular that every clause
        contains at least two variables). Consider 
	the node $u$ of $T$ which witnesses this.
	We can write each clause $D_i$ of $S$ as $X_i \sqcup Y_i$, 
	where $X_i$ is $D_i \cap \LEAVES(T \setminus T_u)$ and $Y_i$ is $D_i \cap \LEAVES(T_u)$.
	Then according to Lemma~\ref{lem:dncpi_fn}, there exists a valuation $\nu$ of the variables of $\phi$ with domain $V \setminus \{x_1,y_1,\ldots,x_m,y_m\}$, where 
	$x_i \in X_i$ and $y_i \in Y_i$ for $1 \leq i \leq m$, such that $\nu(\phi)$ captures the Boolean function $\SINT_m$, hence we know that $\nu(C)$ also
	captures $\SINT_m$.
	But by Proposition~\ref{prop:set_intersection_lower}, we have $|\nu(C)| \geq 2^m -1 \geq 2^n -1$, hence $|C| \geq 2^n -1$.
\end{proof}
\end{toappendix}

This observation allows us to leverage the bound of~\cite{bova2016knowledge}
on the size of d-SDNNFs that compute $\SINT_n$, assuming that they are
structured by an ``inconvenient'' v-tree:

\begin{proposition}[{\cite[Proposition~14]{bova2016knowledge}}]
	\label{prop:set_intersection_lower}
        Let $X_n = \{x_1, \ldots, x_n\}$ and $Y_n = \{y_1, \ldots, y_n\}$ for $n
        \in \NN$, and 
	let $T$ be a v-tree over $X_n \sqcup Y_n$ such that there exists a node
        $n \in T$ with $X_n \subseteq \LEAVES(T_n)$ and $Y_n \subseteq \LEAVES(T \setminus T_n)$.
	Then any d-SDNNF structured by $T$ computing $\SINT_n$ has size $\geq 2^n -1$.
\end{proposition}

In our setting, an ``inconvenient'' v-tree for a dncpi-set is one that shatters
it: each clause of the dncpi-set is then partitioned in two non-empty
subsets where we can pick $x_i$ and $y_i$ for Lemma~\ref{lem:dncpi_fn}.
Hence, when every v-tree shatters a large dncpi-set of~$\phi$,
Proposition~\ref{prop:set_intersection_lower} allows
us to deduce the lower bound on the size of every d-SDNNF for~$\phi$. We have
thus shown Lemma~\ref{lem:dncpi-shattered2},
and this 
concludes the proof of Theorem~\ref{thm:dSDNNFlower} (in the DNF case).

\begin{toappendix}
\subsection{From DNFs to CNFs}
\label{apx:dSDNNFCNFs}
We now argue that Theorem~\ref{thm:dSDNNFlower} also holds for monotone CNFs and SDNNFs.
Note that we cannot use a dualization argument as we did in the previous
  section, as we are now working with DNNFs that are not necessarily deterministic.
Observe that Definition~\ref{def:dncpi} and Lemma~\ref{lem:exists-dncpi2}
  can also apply
to monotone CNFs as these only use the hypergraph corresponding to the
  formula, not the semantics of the formula.
Hence, in order to apply the same arguments as in the DNF case and prove
  an analogue of Lemma~\ref{lem:dncpi-shattered2} in the CNF case, the only
  difference is that we would need to consider the function $f_n \defeq
  (x_1 \lor y_1) \land \ldots \land (x_n \lor y_n)$ and obtain analogues of Lemma~\ref{lem:dncpi_fn}
and Proposition~\ref{prop:set_intersection_lower} for that function.
This is clear for Lemma~\ref{lem:dncpi_fn}, so we only need to check that the analogue of Proposition~\ref{prop:set_intersection_lower} holds.
To understand why, we need to go deeper into the proof
  from~\cite{bova2016knowledge}. They paraphrase a result of
  Pipatsrisawat and Darwiche~\cite{pipatsrisawat2010lower} in the
  following way:

  \begin{theorem}[({\cite[Theorem~13]{bova2016knowledge} and
    \cite[Theorem~3]{pipatsrisawat2010lower}})]
	\label{thm:has_rectangle_cover}
Let $C$ be a SDNNF on variables~$V$ structured by a v-tree $T$, and let $f$ be
    the function that it captures.
For every node $n \in T$, the function $f$ has a rectangle cover of size $\leq |C|$ with partition 
$(V \cap \LEAVES(T_n), V \cap \LEAVES(T \setminus T_n))$.
\end{theorem}

Here, a \emph{rectangle cover} of a Boolean function $f: X\sqcup Y \to \{0,1\}$ with partition $(X,Y)$ is a disjunction $\bigvee\limits_{i=1}^m (g_i(X) \land h_i(Y))$ equivalent to $f$ such that
$g_i$ (resp., $h_i$) is a Boolean function 
  on variables~$X$ (resp., on variables~$Y$), and $m$ is its \emph{size}.
This notion is a standard tool for showing lower bounds in communication complexity.
Therefore, we are interested in the smallest size of a rectangle cover for the function $f_n: X \sqcup Y \mapsto (x_1 \lor y_1) \land \ldots \land (x_n \lor y_n)$ under
partition $(X,Y)$.
But it is known from communication complexity that any rectangle cover for the function \emph{set disjunction} 
$\SDISJ_n : X \sqcup Y \mapsto (\lnot x_1 \lor \lnot y_1) \land \ldots
  \land (\lnot x_n \lor \lnot y_n)$ has size $\geq 2^n$
 (see paragraph ``Fooling set method'', page 5 of~\cite{sherstov2014thirty}).
Moreover, it is easy to see that we can turn any rectangle cover of size $m$ for $f_n$ with partition $(X,Y)$ into a rectangle cover for $\SDISJ_n$ of the same size
and under the same partition, which implies that any such cover for $f_n$ must be of size at least $\geq 2^n$ and concludes the proof. 
Indeed, let $\bigvee\limits_{i=1}^{m} (g_i(X) \land h_i(Y))$ be a rectangle cover for $f_n$ with partition $(X,Y)$.
When $\nu$ is a Boolean valuation from $S$ to $\{0,1\}$, let us write
  $\overline{\nu}$ for the Boolean valuation from~$S$ to~$\{0,1\}$ defined by
  $\overline{\nu}(s) \colonequals 1-\nu(s)$ for $s \in S$.
We then define $\overline{g_i}$ for $1 \leq i \leq n$ (resp., $\overline{h_i}$) to be the Boolean function from $X$ (resp., $Y$) to $\{0,1\}$ defined by
$\overline{g_i}(\nu) \colonequals g_i(\overline{\nu})$ for all valuations $\nu:X
  \to \{0,1\}$ (resp., $\overline{h_i}(\nu) \colonequals h_i(\overline{\nu})$).
One can then check that $\bigvee\limits_{i=1}^{m} (\overline{g_i}(X) \land \overline{h_i}(Y))$ is a rectangle cover for $\SDISJ_n$ of size $m$ with partition $(X,Y)$.
\end{toappendix}

\section{Application to Query Lineages}
\label{sec:lineages}
In this section, we adapt the lower bound of the previous section to
the computation of query lineages on relational instances. Like
in~\cite{amarilli2016tractable}, for technical reasons, we must assume a graph
signature. We first recall some preliminaries and then state our result.

\myparagraph{Preliminaries}
We fix a \emph{graph signature} $\sigma$ of
relation names and arities in~$\{1, 2\}$, with at least one relation of
arity~$2$.
An \emph{instance} $I$ on~$\sigma$ is a finite set of \emph{facts}
of the form $R(a_1, \ldots, a_n)$ for $n$ the arity of~$R$;
we call $a_1, \ldots, a_n$ \emph{elements} of~$I$.
An instance $I'$ is a
\emph{subinstance} of~$I$ if the
facts of~$I'$ are a subset of those of~$I$. The \emph{Gaifman graph} of~$I$
has the elements of~$I$ as vertices and has one edge 
between each pair of elements that co-occur in some fact
of~$I$. The \emph{treewidth} $\tw(I)$ of~$I$ is that of its Gaifman graph.

A \emph{Boolean
conjunctive query} (CQ) is an existentially quantified conjunction of
\emph{atoms} of the form $R(x_1, \ldots, x_n)$
where the $x_i$ are \emph{variables}.
A \emph{UCQ} is a disjunction of CQs, and a $\ucqneq$ also allows atoms of the
form $x \neq y$.
A $\ucqneq$ is \emph{connected} if the Gaifman graph of each disjunct
(seen as an instance, and ignoring $\neq$-atoms) is connected.
For instance, letting~$\sigma_R$ consist of one arity-2 relation~$R$,
the following connected $\ucqneq$ tests
if there are two facts 
that
share one element:
  $Q_\p: \exists x y z ~ (R(x, y) \lor R(y, x))) \land (R(y, z) \lor R(z, y))
  \land x \neq z$.
(While $Q_\p$ is not given as a disjunction of CQs, it can be
rewritten to one using distributivity.)

The \emph{lineage} of a $\ucqneq$ $Q$ over~$I$ is a Boolean
formula $\phi(Q, I)$ on the facts of~$I$
that maps each Boolean valuation $\nu: I \to \{0, 1\}$
to~$1$ or~$0$ depending on whether $I_\nu$ satisfies~$Q$ or not, where
$I_\nu \colonequals \{F \in I \mid \nu(F) = 1\}$.
The lineage intuitively represents which facts of~$I$ suffice to satisfy~$Q$.
Lineages are useful to evaluate queries on \emph{probabilistic databases}~\cite{suciu2011probabilistic}:
we can obtain the probability of the query from an OBDD or d-DNNF
representing its lineage.

\myparagraph{Problem statement}
We study when query lineages can be computed efficiently in
\emph{data complexity}, i.e., as a function of the input instance, with the
query being fixed. 
A first question asks which \emph{queries} have \emph{tractable lineages} on all instances:
Jha and Suciu \cite[Theorem~3.9]{jha2012tractability} showed that
\emph{inversion-free} $\ucqneq$ queries
admit OBDD representations in this sense,
and Bova and Szeider
\cite[Theorem~5]{bova2017circuit}
have recently shown that $\ucqneq$ queries with inversions do not even have tractable
d-SDNNF lineages.
A second question
asks which \emph{instance classes}
ensure that \emph{all queries} have tractable lineages on them.
This was studied for OBDD representations in~\cite{amarilli2016tractable}:
bounded-treewidth instances have tractable OBDD lineage representations for any
MSO query (\cite[Theorem~6.5]{amarilli2016tractable}, using
\cite{jha2012tractability});
conversely
there are \emph{intricate} queries (a class of connected
$\ucqneq$ queries) whose lineages never have tractable OBDD
representations in the instance treewidth~\cite[Theorem~8.7]{amarilli2016tractable}.
The query $Q_\p$ above is an example of an intricate query on the signature
$\sigma_R$ (refer to~\cite[Definition~8.5]{amarilli2016tractable} for the formal definition of intricate queries).
This result shows that we must bound instance treewidth for all queries to have
tractable OBDDs, but leaves the question open for more expressive
lineage representations.

\myparagraph{Result}
Our bound in the previous section allows us to extend Theorem~8.7
of~\cite{amarilli2016tractable} from OBDDs to d-SDNNFs, yielding the
following:

\begin{theoremrep}
  \label{thm:querymain}
  There is a constant $d\in\NN$ such that the following is true.
  Let $\sigma$ be an arity-2 signature, and $Q$ a connected $\ucqneq$ which
  is intricate on~$\sigma$. For any instance $I$ on~$\sigma$, any d-SDNNF representing the
  lineage of~$Q$ on~$I$ has size $2^{\Omega(\tw(I)^{1/d})}$.
\end{theoremrep}

\begin{proofsketch}
  As in~\cite{amarilli2016tractable}, we use a result of Chekuri and 
  Chuzhoy~\cite{chekuri2014polynomial} to
  show that the Gaifman graph of $I$ has a degree-3 topological minor $S$
  of treewidth $\Omega(\tw(I)^{1/d})$ for some constant $d\in\NN$; we also ensure
  that $S$ has sufficiently high \emph{girth} relative to~$Q$.
  We focus on a subinstance $I'$ of~$I$ that corresponds to~$S$: this suffices
  to show our lower bound, because we
  can always compute a tractable representation of $\phi(Q, I')$ from one
  of~$\phi(Q, I)$. Now, we can represent
  $\phi(Q, I')$ as a minimized DNF $\psi$ by enumerating its minimal matches:
  $\psi$ has
  constant arity because the number of atoms of~$Q$ is fixed, and it has
  constant degree because $S$ has constant degree and
  $Q$ is connected. Further, as $Q$ is intricate and $I'$
  has high girth relative to~$Q$, we can ensure that this
  DNF has treewidth $\Omega(\tw(I'))$.
  We conclude by Theorem~\ref{thm:dSDNNFlower}:
  d-SDNNFs representing $\phi(Q, I')$, hence $\phi(Q, I)$, have size
  $2^{\Omega(\tw(I)^{1/d})}$.
\end{proofsketch}

\begin{toappendix}
  In this section, we prove Theorem~\ref{thm:querymain}. We will use
  the restatement of the main result of~\cite{chekuri2014polynomial} given in
  \cite{amarilli2016tractable} (where a \emph{degree-3} graph is one where the
  maximal degree is~3):
  
  \begin{lemma}[{(\cite{chekuri2014polynomial}, rephrased as
    \cite[Lemma~4.4]{amarilli2016tractable})}]
    \label{lem:extraction}
    There is $c \in \NN$ such that, for
  any degree-3 planar graph $H$, for any graph $G$ of treewidth $\geq
    \card{V(H)}^c$, $H$ is a topological minor of~$G$.
  \end{lemma}
 
  We set $d \colonequals 2c$, for the $c$ of this lemma.
  Fix the arity-2 signature $\sigma$ and the intricate query~$Q$.
  As $\sigma$
  is nonempty, the tautological and vacuous $\ucqneq$ queries are not intricate,
  so we can assume that $Q$ is not trivial in this sense.
  We denote by
  $\card{Q}$ the number of atoms of~$Q$.

  We now define the class of subgraphs that we wish to extract.
  Recall that the \emph{girth} of an undirected graph is defined as the length
  of the shortest simple cycle in the graph (or $\infty$ if the graph is
  acyclic).
  Let us define an
  infinite family $\calS = S_2, \ldots, S_n, \ldots$ of graphs to
  extract, such that, for each $i\in\NN$, the graph $S_i$ satisfies the
  following:

\begin{enumerate}
  \item it is a degree-3 graph;
  \item it has treewidth~$i$;
  \item it has $\leq \alpha \times i^2$ vertices for some constant $\alpha \geq 1$ depending only on~$Q$;
  \item it has no vertex of degree 1;
  \item it has girth~$> 2 \card{Q} + 2$;
  \item it is planar.
\end{enumerate}

  We can define each $S_i$ by starting, for instance, with a wall graph
  \cite{dragan2011spanners}, to satisfy the first three conditions (for
  some fixed $\alpha$) as well as condition~6.
  We then iteratively remove all vertices of
  degree~$1$, which clearly does not impact treewidth or planarity. Indeed, treewidth
  cannot increase when we do this, the graph cannot become empty (because its
  initial treewidth is $\geq 2$, so it has a cycle, which will never be
  removed), and treewidth cannot decrease either. Specifically, if we consider a
  graph $G$ and the result $G'$ of removing one vertex $v$ of degree~$1$ in~$G$,
  given a tree decomposition~$T'$ of~$G'$, we can construct a tree
  decomposition $T$ of~$G$ by adding one bag with $v$ and its one incident
  vertex $w$, and connecting it to a bag containing $w$ in~$T'$ (if one
  exists; we connect it arbitrarily otherwise); the result is clearly a tree
  decomposition of~$G$, and the width is unchanged because $G'$ is non-empty
  so the maximal bag size in~$T'$ is $\geq 2$. This satisfies 
  requirement~4 and does not break requirements 1--3 or~6. Last, 
  we subdivide each edge into a path
  of length $2\card{Q} + 3$ to ensure that the girth condition is respected: this
  satisfies requirement~5, does not affect requirements 1--2 or~4 or~6, and
  requirement~3 is still satisfied up to multiplying $\alpha$ by $3\times
  (2 \card{Q}
  + 2)+1$ (each path replacing an edge introduces $(2 \card{Q}
    + 2)$ new vertices, and since the graph is degree-3, an upper bound on the
    number of edges is three times the number of vertices).

  We now make explicit the function hidden in
  the $\Omega$-notation in the exponent of
  the bound that we wish to show. This function will only depend on~$Q$.
  Define the increasing function $f: k \mapsto \frac{1}{\alpha} k^{1/d}
  $, and let $k_0 \in \NN$ be the smallest value of~$k$ such that
  $f(k) \geq 2$. We will show that the size of a d-SDNNF for
  an input instance $I$ is $\geq 2^{\beta f(\tw(I))}$ when $\tw(I)$ is large
  enough, for some constant
  $\beta>0$ to be defined later, depending only on~$Q$.
  This means indeed that it is a $\Omega(\tw(I)^{1/d})$.
  We assume $\tw(I)\geq k_0$ (and thus $f(k)\geq 2$) in what follows.

  Let $I$ be the input instance on~$\sigma$, let $G$ be the Gaifman graph
  of~$I$, and let $k \colonequals \tw(I)=\tw(G)$.
  Let $k' \colonequals f(k)$, 
  and consider $S_{k'}$, which is well-defined because $k'$ is an integer
  which is $\geq 2$.
  We
  know that the number $n_{k'}$ of vertices of~$S_{k'}$ 
  is such that $n_{k'} \leq \alpha k'^2$, i.e., $n_{k'} \leq k^{1/c}$,
  so the treewidth $k$ of~$G$ is $\geq n_{k'}^c$.
  Hence, we know by
  Lemma~\ref{lem:extraction} that $S_{k'}$ is a topological minor of~$G$. Let
  $G'$ be the subgraph of~$G$ corresponding to this topological minor: it is a
  subgraph of~$G$, and a subdivision of~$S_{k'}$.
  
  We will extract a
  corresponding subinstance $I'$ of~$I$ whose Gaifman graph is~$G'$. For
  simplicity, we will ensure that $I'$ is \emph{Gaifman-tight}. An instance
  $I_0$ is \emph{Gaifman-tight} if two conditions hold: first, letting $G_0$ be
  the Gaifman graph of~$I_0$, for each edge $\{a, b\}$ of~$G_0$, there is
  exactly one fact of~$I_0$ containing $a$ and $b$ (hence, of the form $R(a, b)$
  or $R(b, a)$, with $a \neq b$);
  second, every fact of $I_0$ is a binary fact with two distinct elements (of
  the form $R(c, d)$ with $c \neq d$). Intuitively, an instance is Gaifman-tight
  if it is exactly obtained from its Gaifman graph by choosing one relation
  name and 
  orientation for each edge of the Gaifman graph.

  We define a Gaifman-tight subinstance $I'$ of $I$ with Gaifman graph $G'$
  by keeping, for every edge $\{a,
  b\}$ of~$G'$, exactly one binary fact of~$I$ containing the two elements $a$
  and $b$
  (which must exist by definition of the Gaifman graph). By
  construction, the Gaifman graph of~$I'$ is then~$G'$. Hence, we know the
  following about the subinstance $I'$ of~$I$ and its Gaifman graph~$G'$
  (the numbering of this list follows the list of conditions on~$\calS$):

  \begin{enumerate}
    \item For every element $a$ of~$I'$, there are at most
      $3$ facts where $a$ occurs (because $G'$ has maximal degree 3).
\item The treewidth of~$I'$ is $k'$.
\item (N/A: There is no analogue of the requirement~3 imposed on~$\calS$)
\item There are no vertices of degree~$1$ in~$G'$.
\item The girth of~$G'$ is $> 2 \card{Q} + 2$ (because as a subdivision of~$S_{k'}$ its girth
  is at least that of~$S_{k'}$).
\item $G'$ is planar.
\item $I'$ is Gaifman-tight.
  \end{enumerate}

  We will now construct a DNF representation of~$\phi(Q, I')$. Remember that $Q$
  is a $\ucqneq$, so it is monotone, hence we can construct $\phi(Q, I')$
  to be a monotone DNF. As
  $Q$ is not trivial, $\phi(Q, I')$ will contain at least one nonempty clause.
  Further, the DNF can be computed as a minimized DNF by taking the
  disjunction of conjunctions that stand for each \emph{minimal match} of~$Q$
  in~$I'$. Specifically, a \emph{minimal match} of~$Q$ in~$I'$ is a subinstance
  $M$ of~$I'$ such that $M \models Q$ and there is no $M' \subsetneq M$
  such that $M' \models Q$. The following is then easy to see (and this
  monotone DNF representation is clearly unique):
\[
  \phi(Q, I') \colonequals \bigvee_{\substack{M
  \text{~minimal}\\\text{match~of~} Q\\\text{~in~}
  \smash{I'}}} \quad\bigwedge_{F \in M} F
  \]

  Let $H$ be the hypergraph of this DNF. 
  To be able to usefully apply Theorem~\ref{thm:dSDNNFlower}, we must show that
  the arity and degree of~$H$ are constant, and that $\tw(H)$ is
  $\Omega(\tw(I'))$. We
  first show the first claim. The arity of~$H$ is clearly bounded from above by
  the size of a minimal match of~$Q$ in~$I'$, whose size is clearly bounded from
  above by $\card{Q}$, which is constant. As for the degree
  of~$H$, as $Q$ is a connected query, any minimal match of~$Q$ on~$I'$
  involving some fact $F$ must be contained in the subinstance of~$I'$ induced
  by the ball of radius $\card{Q}$ centered around the elements of~$F$ in~$G'$: as the degree
  of~$G'$ is at most~$3$, this ball has constant size, so, as $\sigma$ is fixed,
  $F$ can only occur in constantly many different matches, and the degree is constant.
  We now show that $\tw(H)$ is $\Omega(\tw(I'))$: we show this in the following lemma,
  which captures the essence of intricate queries (namely: under some
  conditions, their lineage never has lower treewidth than the input instance):

  \begin{lemma}
    \label{lem:linpreserve}
    Let $\sigma$ be an arity-2 signature,
    let $Q$ be a connected $\ucqneq$ which is intricate for~$\sigma$, and
    let $I'$ be a Gaifman-tight instance on~$\sigma$
    whose Gaifman graph has no degree-1 vertex and has girth $>2\card{Q} + 2$.
    Then, letting $H$ be the
    hypergraph of the monotone DNF representing $\phi(Q, I')$, we have 
    $\tw(H) \geq \left\lfloor \frac{\tw(I')}{2} \right\rfloor$.
  \end{lemma}
  Let us conclude the proof of Theorem~\ref{thm:querymain} using
  Lemma~\ref{lem:linpreserve}, and show Lemma~\ref{lem:linpreserve} afterwards.
  As the arity and degree of~$H$ are bounded by constants, by 
  Theorem~\ref{thm:dSDNNFlower},
  we know that any d-SDNNF for~$\phi(Q, I')$ has size
  $\geq 2^{\beta'\tw(H)}$ for some constant $\beta' > 0$
  (depending only on the arity and degree bounds on~$I'$ given above, which depend only on~$Q$), which by Lemma~\ref{lem:linpreserve} is
  $\geq 2^{\beta\tw(I')}$ for a different constant $\beta > 0$ and
  $\tw(I')$ large enough. By definition of~$S_{k'}$, we obtain the lower
  bound of $2^{\beta f(k)}$ for $k$ large enough. Now, to conclude, we must show
  that this lower bound also applies to any d-SDNNF for $\phi(Q, I)$. But it is
  clear that, from any d-SDNNF $C$ for $\phi(Q, I)$, we can obtain a d-SDNNF
  $C'$ for $\phi(Q, I')$ which is no larger than $C$ 
  (structured by a v-tree obtained from that of~$C$), simply by
  evaluating to~$0$ all inputs corresponding to facts of $I \setminus I'$.
  Hence, the lower bound also applies to a d-SDNNF for~$\phi(Q, I)$,
  establishing the result of Theorem~\ref{thm:querymain}.

  \bigskip
  All that remains is to show Lemma~\ref{lem:linpreserve}.
  Let us fix the graph signature $\sigma$,
  the connected $\ucqneq$ $Q$ which is intricate for~$\sigma$, and the instance
  $I'$ on~$\sigma$ satisfying the conditions.
  We say that two different facts $R(a, b)$ and $S(c, d)$ of~$I'$ \emph{touch} if
  they share an element, formally, 
  $\card{\{a, b\} \cap \{c, d\}} = 1$: as $I'$ is Gaifman-tight, remember that we
  must have $a\neq b$, $c\neq d$, and $\{a, b\} \neq \{c, d\}$.
  The key for Lemma~\ref{lem:linpreserve} is then captured in the following
  auxiliary claim:

  \begin{claim}
    \label{clm:adjfact}
    Let $F$ and $F'$ be two facts of~$I'$ that touch. Then there is a minimal
    match $M$ of~$Q$ such that $\{F, F'\} \subseteq M$.
  \end{claim}

  \begin{proof}
    Let $G$ be the Gaifman graph of~$I'$.
    Consider the two edges $e$ and $e'$ standing for $F$ and $F'$ in~$G$:
    these edges are incident in~$G$, so we write without loss of
    generality $e = \{u, v\}$ and $e' = \{v, w\}$.
    Fix $n \colonequals \card{Q}$.
    Define a path $\pi=u u_1\dots u_n$ in~$G$ of $\card{Q}$ edges by exploring~$G$ from~$u$:
    initially we are at~$u$ and call $v$ the \emph{predecessor vertex}, and
    whenever we reach some vertex~$x$, we visit a neighbor of
    $x$ which is different from the predecessor of~$x$, and set $x$ to be the
    new predecessor. Such a path
    exists, because this exploration can only get stuck on a vertex of degree~$1$
    (i.e., a vertex that we cannot exit except by going back on its predecessor), and this cannot happen
    by our assumption that $G$ has no vertex of degree~$1$. We define a path
    $\pi'=w w_1\dots w_n$ in~$G$ of $\card{Q}$ edges by exploring from~$w$ with
    predecessor~$v$ in the same way. Now, we consider the path $\rho$
    obtained by concatenating the reverse of 
    $\pi$, $e$, $e'$, and $\pi'$, namely: $\rho: u_n, \ldots, u_1, u, v, w, w_1,
    \ldots, w_n$. We claim that this path is a simple path, i.e., no two
    vertices in the path are the same. Indeed, by definition, no two consecutive
    vertices can be the same in~$\pi$, in~$\pi'$, or in $u, v, w$. Further, two
    vertices separated by one single vertex cannot be the same: this is the case
    in~$\pi$ and~$\pi'$ because we do not go back to the predecessor
    vertex in the exploration, and
    initially we do not go back on~$v$: and for $u$ and $w$ we know that they
    are different because
    $F$ and~$F'$ touch and~$I'$ is Gaifman-tight. Last, 
    two vertices further apart in~$\rho$
    cannot be equal, because otherwise the path $\rho$ would contain a simple
    cycle of~$G$, which would contradict the hypothesis on the girth of~$G$.

    Hence, $\rho$ is a simple path of the Gaifman graph $G$ of~$I'$. Consider the
    sequence of facts $L$ of~$I'$ that witness the existence of each edge of~$\rho$,
    which is unique because $I'$ is Gaifman-tight; in particular we
    choose $F$ and $F'$ as witnesses for $e$ and $e'$.
    Recall now the definition of a line instance, and of a $\ucqneq$ $Q$
    being intricate (Definitions~8.4 and~8.5 of~\cite{amarilli2016tractable}).
    The sequence of facts $L$
    is a line instance, with $\card{L} = 2\card{Q} + 2$, and the two facts
    incident to the middle element are $F$ and $F'$. Hence, the definition  of
    intricate queries ensures that there
    is a minimal match $M$ of~$Q$ on~$L$ that includes both
    $F$ and $F'$. As $L$ is a subinstance of~$I'$, the match $M$ is still a match
    of~$Q$ on~$I'$, and it is still minimal, because any match $M' \subseteq M$
    would also satisfy $M' \subseteq L$ and contradict the minimality of~$M$
    on~$L$. Hence, $M$ is the desired minimal match, which concludes the proof.
  \end{proof}

  We are now ready to prove Lemma~\ref{lem:linpreserve} from
  Claim~\ref{clm:adjfact}, which is the only
  missing part of the proof of Theorem~\ref{thm:querymain}:

  \begin{proof}[Proof of Lemma~\ref{lem:linpreserve}]
    Fix $\sigma$, $Q$, and $I'$, consider the monotone DNF representation
    of $\phi(Q, I')$ and its hypergraph $H$. To show the desired inequality,
    it suffices to show that, from a tree decomposition $T$ of~$H$ where the
    maximal bag size is~$k$, we can
    construct a tree decomposition~$T'$ of~$I'$ whose maximal bag size is no greater than
    $2k$. Let
    $T$ be a tree decomposition of~$H$, and construct $T'$ to have same skeleton
    as~$T$. We define the labeling $\lambda(b')$ of every bag $b'$ of~$T'$ to be
    the set of vertices occurring in the label~$\lambda(b)$ of the corresponding
    bag $b$ of~$T$ (which consists of variables of $\phi(Q, I')$, hence of facts
    of~$I'$): this clearly satisfies the size requirements. We must now
    show that $T'$ is a tree decomposition of~$I'$.

    To show the occurrence requirement, we must show that for every fact $F$
    of~$I'$, there is a bag of~$T'$ containing its two elements. To show this, it
    suffices to show that there is a bag of~$T$ that contains $F$ (as a vertex
    of~$H$). As the
    Gaifman graph of~$I'$ has no vertex of degree~$1$, there must be a fact $F'$
    of~$I'$ that touches $F$, and we can conclude using a consequence of
    Claim~\ref{clm:adjfact}: $F$ must occur in a minimal match of~$Q$
    on~$I'$
    (together with~$F'$, but we do not use this),
    hence it occurs in a clause of~$\phi(Q, I')$, and the occurrence requirement
    on~$T$ ensures that $F$ occurs in a bag of~$T$.

    To show the connectedness requirement, pick an element $a$ of~$I'$. Its
    occurrences in~$T'$ are the union of the occurrences in~$T$ of the facts
    that contain $a$, which are connected subtrees of~$T$ by the connectedness
    requirement of~$T$. Hence, it suffices to show that their union is
    connected. To do this, let us show that for any two facts $F$ and $F'$ that
    contain $a$, then the subtrees $T_F$ and $T_{F'}$ of their occurrences
    in~$T$ necessarily intersect. This is trivial if $F = F'$; now, if $F\neq
    F'$, since $I'$ is Gaifman-tight, the facts $F$ and $F'$ must
    touch in~$I'$ (they cannot share exactly the same elements). Now, we use
    Claim~\ref{clm:adjfact} to conclude that $F$ and $F'$ occur together in a
    minimal match $M$ of~$Q$ on~$I'$. Hence, there is a clause
    of~$\phi(Q, I')$
    which contains both $F$ and $F'$, which ensures that $F$ and $F'$ occur
    together in a bag of~$T$, so $T_F$ and $T_{F'}$ intersect. This shows that
    $T'$ is indeed a tree decomposition of~$I'$, which concludes the proof.
  \end{proof}

  This concludes the proof of Theorem~\ref{thm:querymain}.
\end{toappendix}

To summarize, given an instance family
$\calI$ satisfying the constructibility requirement of Theorem~8.1
of~\cite{amarilli2016tractable}, there are two regimes: (i.) $\calI$ has bounded treewidth and then all
MSO queries have d-SDNNF lineages on instances of~$\calI$ that are computable in
linear time; or (ii.) the treewidth is unbounded and then there are $\ucqneq$ queries
(the 
intricate ones) whose lineages on instances of~$\calI$ have no
d-SDNNF representations 
polynomial in the instance size.

\section{Conclusion}
\label{sec:conclusion}
We have shown tight connections between 
structured circuit classes and width measures on circuits.
We constructively rewrite bounded-treewidth circuits to d-SDNNFs in time linear
in the circuit and singly exponential in the treewidth, and show matching lower
bounds for arbitrary monotone CNFs or DNFs under degree and arity assumptions;
we also show a lower bound for pathwidth and OBDDs.
Our
results have applications to rich query evaluation:
probabilistic query evaluation, computation of lineages, enumeration,
etc.

Our work also raises a number of open questions. First, the d-SDNNF 
obtained in the proof of Theorem~\ref{thm:upper_bound} does \emph{not} respect
the definition of a
\emph{sentential decision diagram} (SDD)~\cite{darwiche2011sdd}.
Can this be fixed, and Theorem~\ref{thm:upper_bound} extended to SDDs?
Or is it impossible, which could solve the open question~\cite{bova2016sdds}
of separating SDDs and
d-SDNNFs?
Second, can we weaken the hypotheses of bounded degree and arity in
Corollaries~\ref{cor:obdddnf} and~\ref{cor:dsdnnfdnf}, and can we rephrase the
latter to a notion of (d-)SDNNF width to match more closely the statement of the former?
Last, Section~\ref{sec:lineages} shows that d-SDNNF
representations of the lineages of 
intricate queries are 
exponential in the treewidth; we conjecture a similar result for pathwidth and
OBDDs, but this would require a pathwidth analogue of the minor extraction
results of~\cite{chekuri2014polynomial}.

\myparagraph{Acknowledgments}
We acknowledge Chandra Chekuri for his helpful comments at
\url{https://cstheory.stackexchange.com/a/38943/}, as well as Florent
Capelli for pointing out the connection to
\cite[Corollary~6.35]{capelli2016structural} and \cite{vatshelle2012new}.

\clearpage
 
\bibliographystyle{plain}
\bibliography{main}

\end{document}